\documentclass[journal,onecolumn]{IEEEtran}
\usepackage{cite}
\usepackage{amsmath,amssymb,amsfonts}
\usepackage{algorithmic}
 \usepackage{booktabs}
\usepackage{graphicx}
\usepackage{textcomp}
\usepackage{amsfonts,amssymb,amsmath}
\usepackage{graphicx,graphics,epsfig,color}
\usepackage{multicol}
\usepackage{float}
\usepackage{subcaption}
\usepackage{comment}
\usepackage{caption}
\usepackage{diagbox}

\usepackage{multirow}

\usepackage{graphicx} 
\usepackage{here}
\usepackage{amsmath}
\usepackage{amsthm}
\usepackage{amssymb}
\usepackage{amsfonts}
\usepackage{graphicx}
\usepackage{subcaption}
\usepackage{xpatch}
\xpatchcmd{\paragraph}{\normalfont}{{\normalfont\bfseries}}{}{}
\usepackage{color}
\usepackage{dsfont}

\usepackage{nicematrix}

\usepackage[ruled,vlined,linesnumbered]{algorithm2e}

\newtheorem{theorem}{Theorem}
\newtheorem{remark}{Remark}
\newtheorem{lemma}{Lemma}
\newtheorem{proposition}{Proposition}

\newtheorem{definition}{Definition}

\newcommand{\R}{{\mathbb{R}}}

\newcommand{\N}{{\mathbb{N}}}

\newcommand{\Pre}{\mathrm{Pre}}

\newcommand{\argmax}{\textrm{arg}\max}

\newcommand{\Reach}{{\mathrm{Reach}}}

\newcommand{\cl}{\mathrm{cl}}

\newcommand{\dom}{\mathop{\rm dom}\nolimits}

\usepackage{hyperref}

\usepackage{enumitem}
\usepackage{balance}

\definecolor{olivegreen}{rgb}{0.14,0.29,0}

\newif\ifitsdraft
\def\itsdraft{\global\itsdrafttrue}

\itsdraft  

\ifitsdraft

\definecolor{gray}{rgb}{0.33,0.4,0.47}

\definecolor{steelblue}{rgb}{0,.42,.7}

\definecolor{britishgreen}{rgb}{0,0.26,0.15}
\definecolor{navyblue}{rgb}{0,0,.8}
\definecolor{olivegreen}{rgb}{0.14,0.29,0}
\definecolor{myred}{rgb}{0.86,0.1,0.16}
    
\newcounter{al}    \newcounter{ss}

    \allowdisplaybreaks
    \pdfoutput = 1
    
     \usepackage{soul} 

\newcommand{\Safe}{\mathrm{Safe}}

\newcommand{\Pers}{\mathrm{Pers}}
\newcommand{\Rec}{\mathrm{Rec}}
\newcommand{\Int}{\mathrm{Int}}

\begin{document}
\title{Specification-aware Robustness Margins for Symbolic Controllers}

\author{Youssef Ait Si, Antoine Girard, and Adnane Saoud
\thanks{Youssef Ait Si and Adnane Saoud are with the College of Computing, University Mohammed VI Polytechnic, Benguerir, Morocco.  (e-mail: \{youssef.aitsi, adnane.saoud\}@um6p.ma) }
\thanks{Antoine Girard is with Universit\'e Paris-Saclay, CNRS, CentraleSup\'elec, Laboratoire des Signaux et Syst\`emes, 91190, Gif-sur-Yvette, France (e-mail: \{antoine.girard@centralesupelec.fr).}
}

\maketitle                              

\begin{abstract}                          

We address the problem of robust controller synthesis for a class of linear temporal logic (LTL) specifications over families of perturbed systems using symbolic control techniques. Given a dynamical system, a specification, and a symbolic controller synthesized using the fixed-point algorithm of the specification, the objective is to find the maximal perturbation we can apply to the system while the system continues to satisfy the same specification under the same controller. We first provide general results, by demonstrating that controllers synthesized based on the symbolic model can be refined back to a perturbed version of the concrete system while preserving their correctness. Focusing on four fundamental temporal logic specifications, namely safety, reachability, persistence, and recurrence, we introduce a general measure of the maximal robustness margin. Then, for each class of specifications, we derive a customized version of the measure and establish the corresponding theoretical guarantees. Importantly, the robustness margin depends explicitly on the sequence of sets generated during the fixed-point computation, allowing for specification-dependent and less conservative bounds compared to generic abstraction-based approaches. The theoretical developments are illustrated on two examples, demonstrating the practical applicability and effectiveness of the proposed approach.

\end{abstract}

\section{Introduction}
The growing complexity of cyber-physical systems (CPS) has highlighted the critical role of formal methods in controller design, particularly for systems governed by complex temporal logic specifications. Ensuring correctness in such systems is essential, as failures can lead to safety risks and performance degradation in applications like autonomous vehicles, robotics, and industrial automation \cite{vargas2021overview,zhang2021grid,tabuada2009verification,saoud2020contract}.
One effective methodology to achieve such guarantees is through Abstraction-Based Controller Design (ABCD) \cite{ReissigWeberRungger17}, \cite{Kyle2018}, \cite{saoud2019compositional}. The process of ABCD begins by discretizing the state and input spaces of the nonlinear continuous system to create a finite-state symbolic abstraction. This step employs reachability analysis techniques \cite{althoff2021set,COOGAN2017254} to approximate the system’s reachable set of each state and input. The construction of such symbolic models is a well-established foundation \cite{zamani2013symbolic,girard2009approximately,Reissig2011computing}. Next, controllers are synthesized on this abstraction using tools from model checking and formal methods to ensure that complex temporal logic specifications are satisfied \cite{baier2008principles}, with correct by constructions controllers. Finally, the discrete controller is refined back into the continuous domain, ensuring that the correctness properties established on the abstraction hold for the original system. The soundness of this entire procedure depends fundamentally on relationships such as alternating simulation or feedback refinement relations that formally link the original concrete system with its discrete abstraction~\cite{ReissigWeberRungger17,calbert2024,DBLP:conf/concur/AlurHKV98,liu2016finite}, thereby preserving correctness guarantees across the abstraction-refinement stage.

Controller synthesis under complex temporal specifications has been extensively studied for dynamical systems subject to uncertainties and disturbances. One line of work investigates robustness through the lens of assume-guarantee contracts, examining how a system's ability to satisfy a specification with respect to assumption's violation \cite{ehlers2014resilience}. A complementary perspective defines robustness in terms of trajectory proximity: perturbed trajectories are required to remain close to their nominal disturbance-free counterparts \cite{tabuada2014towards,rungger2015notion,apaza2024synthesis}. More recently, new resilience metrics~\cite{saoud2025temporal,Aitsi2026,Negar2026} have been introduced for both controlled and autonomous discrete-time systems operating under finite linear temporal logic (LTL$_f$) specifications \cite{LTLf17}. Nevertheless, these approaches are limited in scope: they do not yield deterministic procedures for nonlinear systems and are limited to fragments of LTL$_f$. To address the presence of disturbances and modeling errors, abstraction-based methods often incorporate explicit uncertainty bounds and compute conservative over-approximations of reachable sets under worst-case realizations of the uncertainties \cite{ReissigWeberRungger17}, \cite{Kyle2018,borri2012symbolic}. Along these lines, \cite{liu2016finite} proposes the construction of robust finite abstractions with suitable safety margins, facilitating controller synthesis from temporal-logic specifications despite intersample behavior, imperfect state measurements, and unmodeled dynamics. To reduce the conservatism of these methods, \cite{bai2019incremental,bai2020accurate} develop adaptive abstractions based on state-input dependent disturbance bounds, together with an incremental refinement algorithm that updates only the regions affected by changes in the uncertainty model. 

Existing robust abstraction methods presented above typically rely on a fixed and predefined set of disturbances. In contrast, this paper focuses on synthesizing controllers that ensure a given specification is satisfied under the \textit{maximal admissible disturbance} without requiring prior knowledge of the disturbance set. Indeed, given a system, a specification, and a symbolic controller designed to satisfy that specification for the nominal (disturbance-free) system, this paper addresses the following question: \textit{What is the maximum disturbance that the system can tolerate while still satisfying the specification under the same designed controller?} 

Our work significantly advances the state of the art beyond \cite{aitsi2025symbolic} by introducing a more flexible and less conservative framework for controller synthesis under uncertainty. While \cite{aitsi2025symbolic} relies on ensuring an alternating simulation relation between the abstract and concrete models and computes the maximal admissible perturbation to preserve this relation, its generality leads to overly conservative results. In our approach, by focusing on commonly encountered specification classes such as reachability, safety, persistence, and recurrence, we are able to tolerate significantly larger perturbations while still ensuring the correctness of the synthesized controllers. Moreover, we provide proofs establishing the new robustness margins exceed the previously bounds in \cite{aitsi2025symbolic}. To achieve these results, we consider a discrete-time control system and construct its symbolic abstraction to enable formal controller synthesis. We introduce the notion of a \textit{perturbed symbolic abstraction}, formally defined as the symbolic model of a perturbed system. 
The core of our approach lies in determining the maximal perturbation bound that preserves equivalence between the predecessor operators for the nominal and perturbed abstractions. We seek the maximal perturbation magnitude that preserves this equivalence at each iteration of the fixed-point algorithm~\cite{istratescu2001fixed}, thus ensuring identical outcomes. This equivalence condition is key to preserving correctness across iterations, enables the quantification of the specification-dependent maximum robustness margin, and implies the resulting controllers from the fixed-point algorithm are guaranteed to be functionally equivalent. We then demonstrate how controllers synthesized for the nominal abstraction can be correctly refined to a perturbed concrete system, as long as the perturbations introduced remain below this computed maximal margin bound. This refinement procedure ensures the controller maintains correctness-by-construction for the concrete system. The perturbed symbolic abstraction and its controller synthesis thus serves as an analytical tool for determining exact bound on the robustness margins. Two numerical examples are proposed to show the performance of the proposed approach. The contributions developed in this paper are as listed below.
\begin{enumerate}
    \item A general definition and characterization of the robustness margin for safety, reachability, persistence, and recurrence specifications, computed directly from the sets generated by the corresponding fixed-point synthesis algorithm.
    \item A refinement procedure that maps a controller synthesized on the disturbance-free symbolic abstraction directly to the perturbed continuous system.
    \item Theoretical proofs showing our margins exceed previously known bounds in \cite{aitsi2025symbolic}.
\end{enumerate}   
Although the technical development is presented for the four specifications above, controller synthesis for several important fragments of LTL can be reduced to these classes. Specifically, synthesis for an LTL formula $\varphi$ is typically performed on the product of the system and a deterministic $\omega$-automaton recognizing $\varphi$, where the automaton's acceptance condition determines the resulting specification class: safety-LTL reduces to safety, co-safety-LTL to reachability, LTL formulas recognized by deterministic Büchi automata reduce to recurrence, and those recognized by deterministic co-Büchi automata reduce to persistence~\cite{DBLP:conf/concur/AlurHKV98,baier2008principles,2001automata}. Consequently, the robustness margins developed in this work extend directly to controllers synthesized for these LTL fragments.

\begin{table}[t]
\centering
\caption{Organization of theoretical results across the four temporal specifications.}
\label{tab:paper_structure}
\small
\setlength{\tabcolsep}{6pt}
\begin{tabular}{lcc}
\toprule
\textbf{Specification} & \textbf{Controllers} & \textbf{Refinement} \\
 & \textbf{equivalence} & \textbf{theorem} \\
\midrule
Safety       & Prop.~\ref{prop:safety_equiv}     & Thm.~\ref{thm:safety_refine}    \\
Reachability & Prop.~\ref{prop:reach_equiv}      & Thm.~\ref{thm:reach_refine}     \\
Persistence  & Prop.~\ref{prop:pers_equiv}       & Thm.~\ref{thm:pers_refine}      \\
Recurrence   & Prop.~\ref{prop:recu_equiv}          & Thm.~\ref{thm:recu_refine}       \\
\bottomrule
\end{tabular}
\end{table}
Table~\ref{tab:paper_structure} summarizes the parallel structure of Sections~\ref{sec:4}, \ref{sec:5}, \ref{sec:6} and \ref{sec:7}, each treating one of the four temporal specifications using the general results established in Section~\ref{Sec:3}.
The remainder of the paper is organized as follows.: Section~\ref{sec:2} shows some
required preliminaries on transition systems, control systems, symbolic abstractions, fixed-point and mu-calculus theory. In Section~\ref{Sec:3}, we provide general results, for refinement procedure and robust predecessor operator. In Sections~\ref{sec:4}, \ref{sec:5}, \ref{sec:6}, and~\ref{sec:7}, we address the specifications of safety, reachability, persistence, and recurrence, respectively. For each specification, we present the corresponding fixed-point algorithm, demonstrate how to compute the maximal robustness margins, and show how to refine the associated robust controller for the concrete perturbed system. Section~\ref{sec:9} shows that, the margins obtained in Sections~\ref{sec:4}--\ref{sec:7} are
less conservative than the abstraction-based margin of~\cite{aitsi2025symbolic}. Finally, in Section~\ref{sec:10}, illustrative examples are proposed in order to show the efficiency of the proposed approach.

\section{Preliminaries}
\label{sec:2}
\textbf{Notations:}
The symbols $ \N $, $ \N_{\geq 0} $, $ \R$, and $\R_{\geq 0},$ denote the set of positive integers, nonnegative integers, real and non-negative real numbers, respectively. Given sets $X$ and $Y$, we denote by $f: X \rightarrow Y$ an ordinary map from $X$ to $Y$, whereas $f : X \rightrightarrows Y$ denotes a set valued map. The notation \( 2^{ X} \) represents the power set of \(  X \), which is the set of all possible subsets of \( X \). For a set $X$, $\cl(X)$ denotes its closure, $\Int(X)$ its interior. For $x \in \mathbb{R}^n$, the infinity norm of $x$ is denoted by $\lVert x\rVert$. The closed ball centered at $x \in \mathbb{R}^n$ with radius $r$ is defined by $\mathcal{B}_{r}(x) = \{y \in \mathbb{R}^n |\lVert x- y\rVert \leq r  \}$. The projection map taking $(x, u) \in X \times U $ to $x \in X$ is denoted by $\pi_X$. Let $X^*$ denotes the set of finite sequences of elements in $X$, and $X^w$ denotes the set of infinite sequences of elements in $X$. Given a collection of sets $L = \{X_0, X_1, \dots, X_n \} \subseteq X^{n+1}$, we define the index map $j: X\rightarrow \mathbb{N} \bigcup\{ \infty\}$ as follows: $j_L(x)=\inf \left\{i \in \mathbb{N} \mid x \in X_i\right\}$ with the convention $\inf\{\emptyset\} =\infty$.

\subsection{Transition Systems}

We introduce the concept of a \textit{transition system}~\cite{tabuada2009verification}, which serves as a common modeling framework for both continuous control systems and their symbolic models.

\begin{definition}\label{definition6}
A transition system is a tuple \( S = (X, X^0, U, \Delta) \), where:
\begin{itemize}
    \item \( X \) is the set of states,
    \item \( X^0 \subseteq X \) denotes the set of initial states,
    \item \( U \) is the set of control inputs,
    \item \( \Delta \subseteq X \times U \times X \) defines the transition relation.
\end{itemize}
A transition \( (x, u, x') \in \Delta \) signifies that the state \( x' \) is reachable from the state \( x \) under input \( u \). This can be denoted as \( x' \in \Delta(x, u) \), and we refer to \( x' \) as a \( u \)-successor of \( x \).
\end{definition}

Given a state \( x \in X \), the set of admissible (or enabled) inputs is denoted by \( U^a(x) = \{ u \in U \mid \Delta(x, u) \neq \emptyset \} \). A transition system \( S \) is called \emph{non-blocking} if for every state \( x \in X \), we have \( U^a(x) \neq \emptyset \). Throughout the paper, we will utilize one fundamental operator to formalize state-control pairs that leads to successors within a given target set \( Z \subseteq X \times U \), called predecessor operator \cite{cai1989program}.
\begin{equation}
\label{Def:predecessor}
\Pre_{\Delta}(Z) = \{ (x, u) \in X \times U \mid \emptyset \neq \Delta(x, u) \subseteq \pi_X(Z) \},
\end{equation}
Where  $\pi_X(Z) = \{x \in  X |    \exists u \in U, (x, u) \in Z\}$ is the projection of $Z$ in the space of states $X$ and with the convention that \( \Pre_\Delta(\emptyset \times U) = \emptyset \times U \). A behavior of the transition system \( S = (X, X^0, U, \Delta) \) is a sequence $\sigma = ((x_0, u_0), (x_1, u_1), (x_2, u_2), \ldots, x_N),$ with \( N \in \mathbb{N} \bigcup \{+\infty\} \), such that \( x_0 \in X^0 \), \( u_i \in U^a(x_i) \), and \( x_{i+1} \in \Delta(x_i, u_i) \) for all \( i < N \). The associated state trajectory is \( \sigma_x = (x_0, x_1, \ldots, x_N) \). A behavior is:
\begin{itemize}
    \item \emph{maximal} if it cannot be extended further while preserving the same initial segment,
    \item \emph{complete} if its length is infinite, i.e., \( N = +\infty \).
\end{itemize}
The set of all maximal behaviors of \( S \) is denoted by \( \mathcal{H}(S) \), and the corresponding set of all maximal state trajectories is written \( \mathcal{H}_x(S) \). A specification for \( S \) is a subset of $X^* \cup X^w$, i.e., \( \mathcal{H}_{\text{Spec}} \subseteq X^* \cup X^w \), representing admissible state trajectories. The system satisfies the specification if \( \mathcal{H}_x(S) \subseteq \mathcal{H}_{\text{Spec}} \).

\subsection{Controlled transition systems}
\label{control_system}
We now formalize the notion of a controller for a transition system. Consider the transition system $S(\Sigma) = (X, X^0, U, \Delta)$, and let a controller be defined as a set-valued map $\mathcal{C}: X \rightrightarrows U$, where $\mathcal{C}(x) \subseteq U^a(x)$ for every state $x \in X$. The domain of the controller is denoted by $\mathrm{dom}(\mathcal{C}) = \{ x \in X \mid \mathcal{C}(x) \neq \emptyset \} \subseteq X$, representing the set of states for which control actions are enabled. Based on this, we define the control transition system as the tuple $S^{\mathcal{C}} = (X^{\mathcal{C}}, X^{\mathcal{C}, 0}, U^{\mathcal{C}}, \Delta^{\mathcal{C}})$, where:
\begin{itemize}
    \item $X^{\mathcal{C}} = X \cap \mathrm{dom}(\mathcal{C})$ is the set of admissible states;
    \item  $X^{\mathcal{C}, 0} = X^0 \cap \mathrm{dom}(\mathcal{C})$ denotes the set of initial states;
    \item $U^{\mathcal{C}} = U$ is the input space;
    \item The control transition relation is given by:
$$
x' \in \Delta^{\mathcal{C}}(x, u) \iff x' \in \Delta(x, u) \text{ and } u \in \mathcal{C}(x),
$$

\end{itemize}
meaning that a transition is allowed if and only if the input $u$ is permitted by the controller at state $x$. Given a specification $\mathcal{H}_{Spec}\subseteq X^w \cup X^*$, $\mathcal{C}$ is said to be a controller for the system $S$ and specification $\mathcal{H}_{Spec}$ if $\mathcal{H}_x(S^\mathcal{C}) \subseteq \mathcal{H}_{Spec}$.
\subsection{Control Systems}

Consider a discrete-time control system \( \Sigma \) described by
\begin{equation}
\label{eqn:system}
x_{k+1} = f(x_k, u_k, d_k),
\end{equation}
where \( x_k \in X \subseteq \mathbb{R}^n \) is the state, \( u_k \in U \subseteq \mathbb{R}^p \) the control input, and \( d_k \in D \subseteq \mathbb{R}^q \) a disturbance input. Moreover, we define a perturbed version of this system, denoted \( \Sigma_{\varepsilon} \), as $x_{k+1} \in f(x_k, u_k, d_k) + \mathcal{B}_{\varepsilon(x_k, u_k)}(0),$ where \( \mathcal{B}_{\varepsilon(x_k, u_k)}(0) \) denotes the closed ball centered at zero with radius \( \varepsilon(x_k, u_k) \), and \( \varepsilon: X \times U \to \mathbb{R}_{\geq 0} \). We can model the system \( \Sigma \) as a transition system \( S(\Sigma) = (X, X^0, U, \Delta) \), where:
\begin{itemize}
    \item \( X \subseteq \mathbb{R}^n \) is the state space,
    \item \( X^0 \subseteq X \) the initial states,
    \item \( U \subseteq \mathbb{R}^p \) the control inputs,
    \item \( \Delta(x, u) = f(x, u, D) \) defines the set of possible successors for each \( (x, u) \).
\end{itemize}

Similarly, the perturbed system \( \Sigma_{\varepsilon} \) can be described as a transition system \( S(\Sigma_{\varepsilon}) = (X, X^0, U, \Delta_{\varepsilon}) \), where 
$x' \in \Delta_{\varepsilon}(x, u) \quad \text{if and only if} \quad x' \in f(x, u, D) + \mathcal{B}_{\varepsilon(x, u)}(0).$

\subsection{Symbolic abstraction of dynamical systems}
\label{Def:Sd}
The symbolic abstraction of \( S(\Sigma) \), denoted \( S_d(\Sigma) = (X_d, X_d^0, U_d, \Delta_d) \), is constructed as follows:
\begin{itemize}
    \item \( X_d = \{ q_0, q_1, \ldots, q_N \} \), where \( q_0 = \mathbb{R}^n \setminus X \), and \( \{ q_i \}_{i=1}^N \) is a finite partition of \( X \),
    \item \( X_d^0 \subseteq X_d \) is the set of symbolic initial states,
    \item \( U_d = \{ v_1, \ldots, v_M \} \subseteq U \) is the set of symbolic inputs,
    \item \( \Delta_d \subseteq X_d \times U_d \times X_d \) is the transition relation defined by:
    \begin{enumerate}
        \item For \( q \in X_d \setminus \{q_0\} \), \( q' \in \Delta_d(q, v) \) if \( \cl(\overline{f}(q, v, D)) \cap \cl(q') \neq \emptyset \), where \( \overline{f}(q, v, D) \) is an over-approximation of \( f(q, v, D) \), i.e $f(q, v, D) \subseteq  \overline{f}(q, v, D) $.
        \item \( \Delta_d(q_0, v) = X_d \) for all \( v \in U_d \).
    \end{enumerate}
\end{itemize}
Using this construction of the symbolic model $S_d(\Sigma)$, one can ensure the existence of an alternating simulation relation from \( S_d(\Sigma) \) to \( S(\Sigma) \) \cite{tabuada2009verification}. 

Given a perturbation map \( \varepsilon: X \times U \rightarrow \mathbb{R}_{\geq 0} \), the symbolic model of the perturbed system \( \Sigma_\varepsilon \) is the tuple \( S_d(\Sigma_\varepsilon) = (X_d, X_d^0, U_d,  \Delta_{d, \varepsilon}) \), where \( X_d \), \( X_d^0 \), and \( U_d \) are defined as for the nominal system \( \Sigma \) and the transition relation is defined as follows: 
    \begin{enumerate}
        \item   For $q \in X_d \setminus \{q_0\}$, 
$  q' \in  \Delta_{d, \varepsilon}(q, v) \iff \\ 
 \{\cl\left( \overline{f}(q, v, D) \right) + \mathcal{B}_{\varepsilon_d(q, v)}(0) \}\cap \cl(q') \neq \emptyset,$ where \begin{equation} \label{epsilon_d}
     \varepsilon_d(q, v) = \sup\limits_{x \in q} \varepsilon(x, v) \end{equation} bounds the perturbation over the couple \( (q, v) \). 
    \item For $q = q_0$: $ \Delta_{d, \varepsilon}(q_0, v) = X_d, \quad \forall v \in U_d.$
\end{enumerate}
Note that this symbolic model of the perturbed system will be only used as analytical tools for proving the main results in this work. 

We define a quantizer $Q: X \rightarrow X_d$ that maps each continuous state $x \in X$ to a corresponding discrete (symbolic) state $q \in X_d$, such that $x \in q$. The inverse quantization map $Q^{-1}: X_d \rightrightarrows X$ is formally given by $Q^{-1}(q) = \{ x \in X \mid x \in q \}$ for each state $q \in X_d$. For a discrete subset $A \subseteq X_d$, the inverse image is extended as $Q^{-1}(A) = \bigcup_{q \in A} Q^{-1}(q)$, implying that $x \in Q^{-1}(A)$ if and only if there exists $q \in A$ s.t. $x \in q$. This inverse mapping can also be generalized to trajectory-level specifications. Let $\mathcal{H}_{Spec} \subseteq X_d^w \cup X_d^*$ denote a set of symbolic trajectories satisfying a specification. The corresponding set of continuous trajectories is defined as $Q^{-1}(\mathcal{H}_{Spec}) \subseteq X^w \cup X^*$, where $Q^{-1}(\mathcal{H}_{Spec}) = \{ \sigma_x = (x_0, x_1, \ldots) \in X^w \cup X^* \mid\ 
(Q(x_0), Q(x_1), \ldots) \in \mathcal{H}_{Spec} \}.$






We define the maximal robustness margin by the map \(\eta: X_d \times U_d \times 2^{X_d} \rightarrow \mathbb{R}_{\geq 0},\)
quantifying the largest admissible perturbation for a symbolic state-control pair $(q, v) \in X_d \times U_d$ such that all its successors remain within the interior of a target set $A \subseteq X_d$. Formally, it is defined as
\small
\begin{equation}
\label{Def:eta_a}
\eta(q, v, A)= 
\begin{cases}
\begin{aligned}[t]
 &  \sup \left\{ \varepsilon \geq 0 \,\middle|\,  \cl(\overline{f}(q,v,D)) + \mathcal{B}_{\varepsilon}(0) \subseteq \right. \\
 & \hspace{2.8cm}\left. \Int(Q^{-1}(A)) \right\}  \\
&  \qquad \text{if } \cl(\overline{f}(q,v,D)) \subseteq \Int(Q^{-1}(A)),  \\
& 0  \qquad \text{otherwise}.
\end{aligned}
\end{cases}
\end{equation}



We now introduce the maximal robustness margin over a collection of sets $L = \{Z_0, Z_1, \dots, Z_m\}$, where $Z_i \subseteq X_d \times U_d$ for $i \in \{0,1,\ldots,m\}$. This margin is defined by the map $\gamma: X_d \times U_d \times (2^{X_d \times U_d})^m \rightarrow \mathbb{R}_{\geq 0}$, which assigns to each $(q,v) \in X_d \times U_d$ and $L$, the value
\begin{equation}
\label{Def:gamma_l}
\gamma(q, v, L)  = \eta(q, v, I(q, v, L)),
\end{equation}
with 
 \begin{equation*}
\begin{aligned}
    I(q,v, L)&=     \bigcap\limits_{i=0}^m \{\pi_X(Z_i) \mid Z_i \in L,\; (q, v) \in Pre_{\Delta_d}(Z_i)  \}, \\
 &=  \bigcap\limits_{i=0}^m \{\pi_X(Z_i)  \mid Z_i \in L,\;\Delta_d(q, v) \subseteq \pi_X(Z_i)  \} \subseteq X_d.
\end{aligned}
\end{equation*}
The map $\gamma$ provides a unified measure based on the map $\eta$, but considers the intersection of all sets \(\pi_X(Z_i)\) with $Z_i$ in \(L\) and for which the successors of $(q, v)$ belongs to $Z_i$. This construction allows us to provide a common robustness margin for a collection of sets $L$. Thus, the map $\gamma$ seeks the maximum margins expanding the reachable set of $(q, v)$ to remain inside the intersection of all the sets containing the successors of $(q, v)$. The intersection $I(q,v, L)$ ensures only the minimal zone of sets containing the successors are considered. If the sets $Z_i$ are ordered, i.e., $Z_i \subseteq Z_{i+1}$, then $I(q, v, L)$ corresponds to the smallest set $Z_i$ that contains the successors of $(q, v)$. 

Before presenting our main results, we first introduce the fundamental concept of fixed-point theory and the $\mu$-calculus~\cite{tabuada2009verification,2001automata}. 

\subsection[fixed-points and mu calculus]{fixed-points and $\mu$-calculus}
\label{Sec:fixed_point}

In the study of algorithms characterized by fixed-point expressions, a central concept is monotone functions. Given the  transition system $S(\Sigma)=(X,X^0,U,\Delta)$ in~\eqref{eqn:system}, a map $G:2^{X\times U} \to 2^{X \times U}$ is said to be monotone if for all $Z_1,Z_2 \subseteq X\times U$, with $Z_1 \subseteq Z_2$, $G(Z_1) \subseteq G(Z_2)$. A subset \( Z^* \subseteq X \times U \) is considered a fixed-point of \( G \) if \( Z^* = G(Z^*). \) Among the fixed-points of $G$, two particular types are of special interest: maximal and minimal fixed-points. 
A fixed-point $\hat{Z}$ of $G$ is called the \emph{maximal fixed-point} if, for any $Z \subseteq X \times U$, \(\; \;Z \subseteq G(Z) \implies Z \subseteq \hat{Z}.\)
Conversely, a fixed-point $\check{Z}$ of $G$ is called the \emph{minimal fixed-point} if, for any $Z \subseteq X \times U$, \(\; \;Z \supseteq G(Z) \implies Z \supseteq \check{Z}.\)
To denote these fixed-points, we use \( \nu \) for the maximal fixed-point and \( \mu \) for the minimal fixed-point, defined as follows:
 $\hat{Z} = \nu Z . G(Z),$ and $\check{Z} = \mu Z . G(Z) .$
The computation of the maximal fixed-point can be achieved through non-growing iterations as follows: $\nu^0 Z . G(Z) = X \times U$ and $\nu^{i+1} Z . G(Z)  = G(\nu^i Z . G(Z)).$ In this iteration process, \( \nu^i Z . G(Z) \) represents the result after \( i \) iterations, starting from the initial set \( X \times U\). Similarly, the minimal fixed-point is computed using non-shrinking iterations $\mu^0 Z . G(Z)  = \varnothing \times U $ and $\mu^{i+1} Z . G(Z)  = G(\mu^i Z . G(Z)).$ Here, \( \mu^i Z . G(Z) \) denotes the result after \( i \) iterations, starting from the set \( \varnothing \times U\). The following theorem from~ \cite{2001automata,davey2002introduction} states the existence and uniqueness of minimal and maximal fixed-points for monotone functions over finite domains.
\begin{theorem}[\cite{2001automata,davey2002introduction}]
\label{theorem:existence_fixed_point}    
Let $G: 2^{X \times U} \to 2^{X \times U}$ be a monotone function with $X$ and $U$ finite sets. Then there exist integers $n_1, n_2 \in \mathbb{N}$ such that \(
\nu^{n_1+1} Z . G(Z) = \nu^{n_1} Z . G(Z) 
\quad \text{and} \quad 
\mu^{n_2+1} Z . G(Z) = \mu^{n_2} Z . G(Z).
\)
Moreover, the fixed-points are unique and satisfy \(
\nu Z . G(Z) = \nu^{n_1} Z . G(Z) 
\quad \text{and} \quad 
\mu Z . G(Z) = \mu^{n_2} Z . G(Z).
\)
\end{theorem}

The next section presents general theoretical results applicable to any specification. 

\section{General results for robustness margins }
\label{Sec:3}
In this section, we present two general results that are independent of the specification. 
First, we establish a refinement procedure that lifts a controller synthesized for the 
abstract nominal system to a controller for the perturbed concrete system under any specification. This refinement is a key step to ensures correctness-by-construction is preserved 
when transferring the controller to the perturbed concrete system. 
Compared with existing refinement frameworks based on simulation and approximate 
bisimulation relations \cite{tabuada2009verification,Girard2007,zamani2014symbolic}, 
our approach explicitly characterizes the perturbation bounds required for correctness preservation. Second, considering two predecessor operators defined for the nominal and perturbed symbolic models, 
we derive conditions under which the perturbation guarantees equality of the two predecessor operators. 

\subsection{Refinement procedure}

In the next lemma, we establish a result showing that the behavior of the perturbed concrete system is contained in that of the perturbed symbolic system, under appropriate conditions on the respective perturbation maps.

\begin{lemma}
\label{Proposition:behavior}

Consider the discrete-time control system \(\Sigma\) in~\eqref{eqn:system} and its symbolic model \(S_d(\Sigma) = (X_d, X_d^0, U_d, \Delta_d)\) constructed in Section~\ref{Def:Sd}. Let the maps \(\varepsilon: X \times U_d \to \mathbb{R}_{\geq 0}\) and \(\mu: X \times U_d \to \mathbb{R}_{\geq 0}\) satisfying \(\mu(x, u) \le\varepsilon(x, u)\) for all \((x, u) \in X \times U_d\). The systems \(S(\Sigma_{\mu})\) and \(S(\Sigma_{\varepsilon})\) represent perturbed versions of the system \(\Sigma\). Let \(S_d(\Sigma_{\varepsilon}) = (X_d, X_d^0, U_d,  \Delta_{d, \varepsilon})\) denote the symbolic model of the perturbed system \(\Sigma_{\varepsilon}\), and \(\mathcal{C}_{d, \varepsilon}\) be a controller for this system. We define a controller $\mathcal{C}:X\rightrightarrows U$ for the system \( S(\Sigma_{\mu}) \) as $\mathcal{C}(x) = \mathcal{C}_{d, \varepsilon}(Q(x))$ for $x \in X$. 
For any behavior $\sigma = ((x_0,u_0), (x_1,u_1), \ldots) \in \mathcal{H}(S^{\mathcal{C}}(\Sigma_{\mu}))$ we have that \(\sigma_d = ((q_0,u_0), (q_1,u_1), \ldots) \in \mathcal{H}(S^{\mathcal{C}_{d,\varepsilon}}_{d}(\Sigma_{\varepsilon}))\), with $q_i = Q(x_i)$ for all $i \in \N$.
\end{lemma}


We derive the following auxiliary proposition, which enables the refinement of a controller synthesized for the disturbance-free abstraction into a controller for the concrete perturbed system.

\begin{proposition}
\label{Prop:refinement}
Consider the discrete-time control system $\Sigma$ and its perturbed version $\Sigma_\varepsilon$ in~\eqref{eqn:system}, for a perturbation map $\varepsilon: X \times U_d \rightarrow \mathbb{R}_{\geq 0}$ along with their corresponding symbolic models $S_d(\Sigma) = (X_d, X_d^0, U_d, \Delta_d)$ and $S_d(\Sigma_\varepsilon) = (X_d, X_d^0, U_d,  \Delta_{d, \varepsilon})$, constructed in Section~\ref{Def:Sd}. Let $\mathcal{C}_d$ and $\mathcal{C}_{d,\varepsilon}$ denote controllers synthesized for the symbolic models $S_d(\Sigma)$ and $S_d(\Sigma_\varepsilon)$, respectively, with respect to the same specification $\mathcal{H}_{Spec} \subseteq X_d^w \cup X_d^*$. We define the controller $\mathcal{C}: X\rightrightarrows U_d$ by $\mathcal{C}(x) = \mathcal{C}_{d}(Q(x)).$ If the following holds: $\forall q \in X_d, \quad \mathcal{C}_d(q) = \mathcal{C}_{d, \varepsilon}(q),$ then \(\mathcal{C}\) is a controller for the perturbed system \(S(\Sigma_{\mu})\) and the specification $Q^{-1}(\mathcal{H}_{Spec})$, for any map $\mu: X \times U_d \rightarrow \mathbb{R}_{\geq 0}$ satisfying $\mu(x, u) \le \varepsilon(x, u)$ for all $(x, u) \in X \times U_d$.
\end{proposition}


\subsection{Robust predecessor operator}

In this section, we characterize a bound on the system’s perturbation map for the predecessor operator, defined in~\eqref{Def:predecessor}, to preserve equality across a predefined list of sets, irrespective of whether the nominal or perturbed system is considered.


\begin{theorem}
\label{Thrm:Pre_mu}
Consider the discrete-time control system $\Sigma$ and its perturbed version $\Sigma_\varepsilon$ in~\eqref{eqn:system}, for a perturbation map $\varepsilon: X \times U_d \rightarrow \mathbb{R}_{\geq 0}$ along with their corresponding symbolic models $S_d(\Sigma) = (X_d, X_d^0, U_d, \Delta_d)$ and $S_d(\Sigma_\varepsilon) = (X_d, X_d^0, U_d,  \Delta_{d, \varepsilon})$, constructed in Section~\ref{Def:Sd}. Let $ L = \{Z_0, \dots, Z_m\} \subseteq (X_d \times U_d)^{m+1} $ a collection of sets. If the map $\varepsilon_d: X_d \times U_d \rightarrow \mathbb{R}_{\geq 0}$ as constructed in~\eqref{epsilon_d} satisfies for all $ (q, v) \in X_d  \times U_d, \quad \varepsilon_d(q, v) < \gamma(q, v, L),$ then we have for all $j \in \{0,1, 2, \dots , m\}$, $ \Pre_{\Delta_d}(Z_j) = \Pre_{ \Delta_{d, \varepsilon}}(Z_j).$
\end{theorem}

We now prove that exceeding the bound defined by the map $\gamma$ violates the equality of the predecessors.

\begin{theorem}
\label{Thrm_pre_mu_neq}
In the setting of Theorem~\ref{Thrm:Pre_mu}, let $L = \{Z_0, \dots, Z_m\} \subseteq (X_d \times U_d)^{m+1}$, and assume there exists $(x^*, v^*) \in X \times U_d$ with $q^* = Q(x^*)$ such that$\varepsilon_d(q^*, v^*) > \gamma(q^*, v^*, L). $Then $\mathrm{Pre}_{\Delta_{d,\varepsilon}}(Z_{j^*}) \neq \mathrm{Pre}_{\Delta_d}(Z_{j^*})$ for some $j^* \in \{0, \dots, m\}$.
\end{theorem}

\begin{remark}
Theorems~\ref{Thrm:Pre_mu} and ~\ref{Thrm_pre_mu_neq} together show that $\gamma(\cdot, \cdot, L)$ is the sharp threshold for preservation of the predecessor operator across $L$: the equality $\mathrm{Pre}_{\Delta_d}(Z_j) = \mathrm{Pre}_{\Delta_{d,\varepsilon}}(Z_j)$ holds for every $Z_j \in L$ when $\varepsilon < \gamma$, and fails on at least one $Z_j$ as soon as $\varepsilon$ strictly exceeds $\gamma$ at any state-input pair. Consequently, having a system with perturbations exceeding $\gamma$  means we can no longer satisfy the refinement procedure in Proposition~\ref{Prop:refinement} and thus we cannot guarantee the correctness of the controller for the perturbed concrete system. 
\end{remark}


In the next sections, we follow the procedure illustrated in Fig.~\ref{fig:schema_of_the_paper}. We start from the disturbance-free system $S(\Sigma)$ and construct its symbolic model $S_d(\Sigma)$. We then synthesize a controller $\mathcal{C}_d$ for a given LTL specification. Using the same procedure we  theoretically construct a controller $\mathcal{C}_{d,\varepsilon}$ for the perturbed symbolic model $S_d(\Sigma_\varepsilon)$ for the same LTL specification. The next step consists in quantifying the bound on the perturbation map to preserve the equality of the predecessor operator and thus the equality $\mathcal{C}_{d,\varepsilon} = \mathcal{C}_d$. Then using Proposition~\ref{Prop:refinement}, we show that the controller $\mathcal{C}_d$ can be refined to the perturbed concrete system $S(\Sigma_\varepsilon)$ under the same condition on the perturbation.
In summary, this procedure allows us to determine, for a given controller synthesized for the system to satisfy a specification, the maximal robustness margin the system can tolerate while still satisfying the specification under the same controller.
\begin{figure}[htbp]
    \centering
    \includegraphics[width=0.5\textwidth]{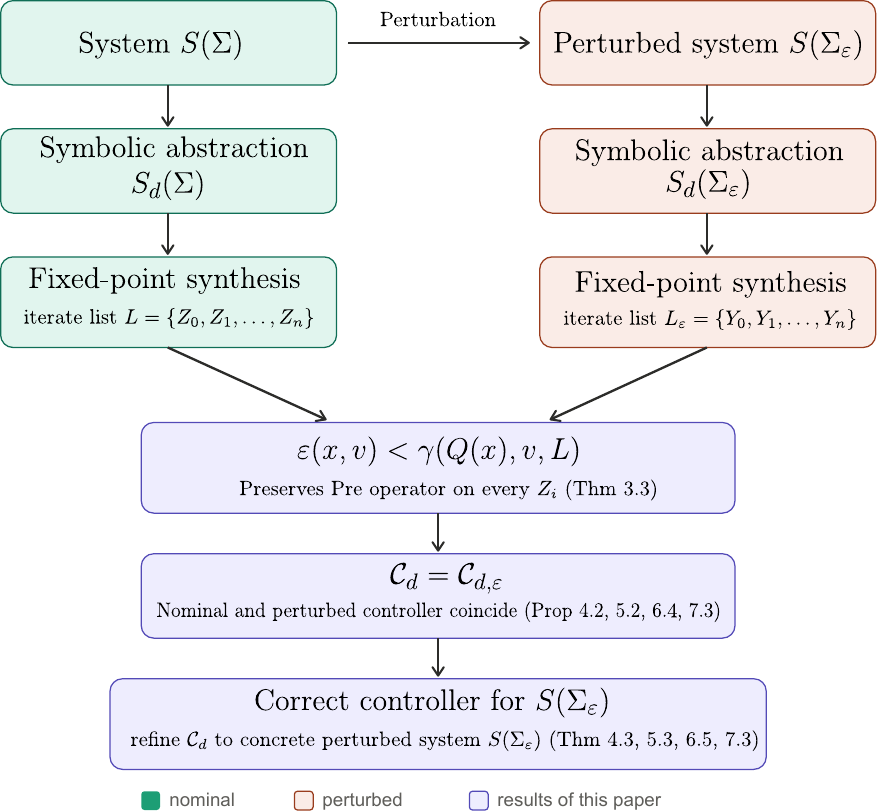}
\caption{Workflow for abstraction based  robust control synthesis. We synthesis the controller $\mathcal{C}_d$ for the nominal system and show the same controller can be refined for perturbed system under conditions on the perturbation. }
    \label{fig:schema_of_the_paper}
\end{figure}

\section{Safety controller}
\label{sec:4}

Consider the system $S(\Sigma) = (X, X_0, U, \Delta)$ in~\eqref{eqn:system}, the safety synthesis problem seeks a  controller $\mathcal{C}$ such that all closed-loop behaviors remain within a prescribed safe set $X_T \subseteq X$. This can be formally defined as follows. 

\begin{definition}
\label{Def:safety_controller}
Consider the system $S(\Sigma) = (X , X_0, U, \Delta)$ in~\eqref{eqn:system},  a controller $\mathcal{C} : X \rightrightarrows U$ and its domain $\dom(\mathcal{C})$. The controller $\mathcal{C}$ is a safety controller for the safe set $X_T \subseteq X$ if for any initial state $x_0 \in \dom(\mathcal{C})$, all maximal trajectories $\sigma=((x_0,u_0),(x_1,u_1),(x_2,u_2)\ldots) \in \mathcal{H}(S^{\mathcal{C}})$ are complete and satisfy, for all $ k \in \mathbb{N}, \quad \sigma_{x}(k)=x_k \in X_T.$
We denote the corresponding safety specification as $\mathcal{H}_{Spec} = \mathcal{H}_{Safety}^{X_T}= \{\sigma_x \in  X^w  \mid  \forall  k \in \mathbb{N},  \sigma_{x}(k)=x_k \in X_T\}.$
    
\end{definition}
The classical approach to synthesize a safety controller is the safety fixed-point Algorithm~\ref{alg:safe}. This algorithm takes as input a transition system \( S(\Sigma) \) and a safe set \( X_T \subseteq X \), and computes the set \( \Safe(X_T) \) and a sequence \(L=\{Z_i \subseteq X \times U\mid i \in \N \}\) containing the set of states and their admissible inputs computed at each iteration of the algorithm. The set  \( \Safe(X_T) \subseteq X \times U\) consists of all state-input pairs \( (x, u) \) such that every possible successor state under input \( u \) remains within \( X_T \). Formally, \( (x, u) \in \Safe(X_T) \) if and only if $\Delta(x, u) \subseteq \pi_X(\Safe(X_T))\subseteq X_T,$ where \( \pi_X(\Safe(X_T)) \) denotes the projection of \( \Safe(X_T) \) on the state space \( X \). 

Note that we have outlined a general structure for this algorithm, regardless whether the system being symbolic or concrete. However, the algorithm operates on the symbolic model $S_d(\Sigma)=(X_d, X_d^0, U_d, \Delta_d)$ constructed in Section~\ref{Def:Sd}, where the domain \( X_d \times U_d \) is finite.


\begin{algorithm}
\caption{Computation of $\Safe(X_T)$}
\label{alg:safe}
\begin{algorithmic}[1]
\STATE \textbf{Input:} $S(\Sigma)$, $W = X_T \times U$ where  $ X_T \subseteq X$
\STATE \textbf{Output:} $\Safe(X_T), L=\{Z_i \mid i \in \N \}$
\STATE $Z_0 := X \times U$
\REPEAT
    \STATE $Z_{i+1} := W \cap \Pre_{\Delta}(Z_i)$
\UNTIL{$Z_{i+1} = Z_i$}
\STATE $\Safe(X_T) := Z_i$
\end{algorithmic}
\end{algorithm}
The safety controller synthesis for the symbolic system $S_d(\Sigma)$ relies on the fixed-point computation of the monotone operator $G: 2^{X_d \times U_d} \rightarrow 2^{X_d \times U_d}$, defined as
$G(Z) = \mathrm{Pre}_{\Delta_d}(Z) \cap W,$ where $W = X_T \times U_d,$
and $X_T \subseteq X_d$ is the discrete safe set. The iterates $Z_i = \nu^i Z.\, G(Z)$, collected in the sequence $L = \{Z_0, Z_1, \ldots, Z_n\}$, converge in finite time to the maximal fixed-point $\mathrm{Safe}(X_T) = Z_n = \nu Z.\, G(Z)$, as established in Theorem~\ref{theorem:existence_fixed_point}. The safety controller $\mathcal{C}_d : X_d \rightrightarrows U_d$ is then constructed as follows:
\begin{equation}
\label{safety_controller_1}
\mathcal{C}_d(q) =
\begin{cases}
\begin{aligned}[t]
&\{ v \in U_d \mid (q, v) \in \nu Z . G(Z) \} \\ 
&\hspace{1.9cm} \text{if } q \in \pi_X(\nu Z . G(Z)), \\
&\emptyset \hspace{1.6cm}  \text{otherwise}.
\end{aligned}
\end{cases}
\end{equation}
This construction yields a correct-by-construction controller for the symbolic abstraction of the safety specification $\mathcal{H}_{\text{Safety}}^{X_T}$ \cite{tabuada2009verification}.
To account for perturbations, consider a function $\varepsilon: X \times U_d \rightarrow \mathbb{R}_{\geq 0}$ characterizing uncertainty, and let $S_d(\Sigma_\varepsilon) = (X_d, X_d^0, U_d,  \Delta_{d, \varepsilon})$ denote the symbolic abstraction of the perturbed system defined in Section~\ref{Def:Sd}. Applying the same fixed-point procedure to $S_d(\Sigma_\varepsilon)$ using the monotone operator $G_\varepsilon(Z) = \mathrm{Pre}_{\Delta_{d,\varepsilon}}(Z) \cap W$ yields the maximal fixed-point $\nu Y . G_\varepsilon(Y)$ and the corresponding sequence $L_\varepsilon = \{Y_0, Y_1, \dots, Y_m\}$. The safety controller for the perturbed abstraction is then given by \begin{equation}
\label{safety_controller_pert}
\mathcal{C}_{d, \varepsilon}(q) =
\begin{cases}
\begin{aligned}[t]
&\{ v \in U_d \mid (q, v) \in \nu Y . G_\varepsilon(Y) \} \\
& \hspace{1.9cm} \text{if } q \in \pi_X(\nu Y . G_\varepsilon(Y)), \\
&\emptyset  \hspace{1.6cm} \text{otherwise}.
\end{aligned}
\end{cases}
\end{equation}

In the subsequent result, we derive a sufficient condition on the perturbation map $\varepsilon$ under which the two controllers $\mathcal{C}_d$ and $\mathcal{C}_{d,\varepsilon}$ are guaranteed to be equivalent.

\begin{proposition}
\label{prop:safety_equiv}
Consider the discrete-time control system $\Sigma$ and its perturbed version $\Sigma_\varepsilon$ defined in~\eqref{eqn:system}, along with their corresponding symbolic models $S_d(\Sigma) = (X_d, X_d^0, U_d, \Delta_d)$ and $S_d(\Sigma_\varepsilon) = (X_d, X_d^0, U_d,  \Delta_{d, \varepsilon})$, constructed in Section~\ref{Def:Sd}. The perturbation map is given by $\varepsilon: X \times U_d \rightarrow \mathbb{R}_{\geq 0}$ and $X_T \subseteq X_d $ is the safe set. Let $\mathcal{C}_d$ and $\mathcal{C}_{d,\varepsilon}$ denote the safety controllers constructed in~\eqref{safety_controller_1} and~\eqref{safety_controller_pert} with the associated sequences $L = \{Z_0, Z_1, \ldots, Z_n\}$ and $L_\varepsilon = \{Y_0, Y_1, \ldots, Y_m\}$. If the perturbation map $\varepsilon: X \times U_d \rightarrow \mathbb{R}_{\geq 0}$ satisfies $\sup_{z\in Q(x)}\varepsilon(z,v) < \gamma(Q(x), v, L)$ for all $(x,v) \in X \times U_d$ then the following holds: $\forall q \in X_d, \quad \mathcal{C}_d(q) = \mathcal{C}_{d, \varepsilon}(q).$
\end{proposition}

In this theorem, we have demonstrated that the controller of the symbolic model is equal to the controller of the perturbed symbolic model, provided that the perturbation is lower than a threshold map. This non-uniform maximum robustness margin map corresponds to the map $\gamma$ defined in~\eqref{Def:gamma_l} that can be explicitly computed using the output list containing the sets generated during the safety algorithm's iterations. In the next theorem which constitutes the main results for this subsection, we show how to refine this controller for the concrete perturbed system.

\begin{theorem}
\label{thm:safety_refine}
Consider the discrete-time control systems $\Sigma$ and its perturbed version $\Sigma_\varepsilon$ in~\eqref{eqn:system}, where \(\varepsilon: X \times U_d \rightarrow \mathbb{R}_{\geq 0}\) and the symbolic model $S_d(\Sigma)=(X_d, X_d^0, U_d, \Delta_d)$ constructed in Section~\ref{Def:Sd}. Let $\mathcal{C}_d$ denote the safety controller constructed in~\eqref{safety_controller_1} with the associated sequence $L = \{Z_0, Z_1, \ldots, Z_n\}$ and $X_T\subseteq X_d$ the safe set. 
We define the controller $\mathcal{C}: X\rightrightarrows U_d$ by $\mathcal{C}(x) = \mathcal{C}_{d}(Q(x))$. For any map \(\varepsilon: X \times U_d \rightarrow \mathbb{R}_{\geq 0}\) that satisfies $\sup_{z\in Q(x)}\varepsilon(z,v) < \gamma(Q(x), v, L)$ for all \((x, u) \in X \times U_d\), the controller \(\mathcal{C}\) is a safety controller for the concrete perturbed system \(S(\Sigma_{\varepsilon})\) and the safety specification $Q^{-1}(\mathcal{H}_{Safety}^{X_T})$.
\end{theorem}

In the following, we present a simplified expression of the admissible robustness margin for safety specifications. 

\begin{proposition}
\label{prop:Safety_explicit}
Consider the discrete-time control system $\Sigma$ in~\eqref{eqn:system}, its symbolic model $S_d(\Sigma) = (X_d, X_d^0, U_d, \Delta_d)$ constructed in Section~\ref{Def:Sd}, and the safety controller $\mathcal{C}_d$ defined in~\eqref{safety_controller_1} with associated iterate sequence $L = \{Z_0, Z_1, \ldots, Z_n\}$, where $Z_n = \nu Z.\, G(Z)$.
If $(x, v) \in  Q^{-1}(Z_n)$, then $\gamma(Q(x), v,L) = \eta(Q(x), v, \dom(\mathcal{C}_d))$. 
\end{proposition}

Proposition~\ref{prop:Safety_explicit} gives a simple and interpretable form
for the safety margin: the intersection over the entire iterate list $L$ in
the definition of $\gamma$ collapses to the single fixed-point set $Z_n$, so
the margin reduces to $\eta(Q(x),v,\dom(\mathcal{C}_d))$. The admissible
perturbation at a state is thus governed only by the distance from its
over-approximated reachable set to the boundary of the safe domain. Specifically, for any state within the safety set (the controller’s domain), any perturbation that keeps the successor states inside this safety set can be tolerated.


The reachability, persistence, and recurrence subsections that
follow adopt the same organization in this subsection. 


\section{Reachability controller}
\label{sec:5}
Let $S = (X, X_0, U, \Delta)$ be the transition system
in~\eqref{eqn:system} and let $X_T \subseteq X$ be a target
set. The reachability synthesis problem consists of finding a
controller $\mathcal{C}$ such that every behavior of the
closed-loop system $S^\mathcal{C}$ reaches $X_T$ in finite time.
No requirement is imposed on the system's behavior thereafter.
\begin{definition}
Consider the system $S(\Sigma)$ in~\eqref{eqn:system},  a controller $\mathcal{C} : X \rightrightarrows U$ and its domain $\dom(\mathcal{C})$. The controller $\mathcal{C}$ is a reachability controller for the target set $X_T \subseteq X$ if for any initial state $x_0 \in \dom(\mathcal{C})$, for all maximal trajectories $\sigma=((x_0,u_0),(x_1,u_1)\ldots )$ of $\mathcal{H}(S^\mathcal{C})$ there exist
$  k \in \mathbb{N}$ such that $\sigma_{x}(k)=x_k \in X_T$. We denote the corresponding reachability specification as $\mathcal{H}_{Spec} = \mathcal{H}_{Reach}^{X_T} = \{\sigma_x \in X^* \cup X^w\mid  \exists  k \in \mathbb{N}, \quad \sigma_{x}(k)=x_k \in X_T\}.$
\end{definition}
To synthesize the reachability controller we rely on Algorithm~\ref{alg:reach}. This algorithm computes \(\Reach(X_T)\), which includes all states, with the corresponding admissible inputs, from which one can reach $X_T$. Furthermore, the algorithm records the sequence of state-input pairs \(L=\{Z_i \subseteq  X \times U\mid i \in \N \}\) computed at each iteration of the algorithm.


\begin{algorithm}
\caption{Computation of $\Reach(X_T)$}
\label{alg:reach}
\begin{algorithmic}[1]
\STATE \textbf{Input:} $S$, $W = X_T \times U$ where  $ X_T \subseteq X$
\STATE \textbf{Output:} $\Reach(X_T), L=\{Z_i \mid i \in \N \}$
\STATE $Z_0 := \emptyset \times U$
\REPEAT
    \STATE $Z_{i+1} := W \cup \Pre_{\Delta}(Z_i)$
\UNTIL{$Z_{i+1} = Z_i$}
\STATE $\Reach(X_T) := Z_i$ 
\end{algorithmic}
\end{algorithm}

Let \( S_d(\Sigma) = (X_d, X_d^0, U_d, \Delta_d) \) denote the symbolic abstraction of \( \Sigma \) constructed in Section~\ref{Def:Sd}. The monotone operator associated with the reachability specification is \(G(Z) = \Pre_{\Delta_d}(Z) \cup W, \quad \text{with } W = X_T \times U_d.
\) The fixed-point iterations (cf. Section~\ref{Sec:fixed_point}) are given by $
Z_i = \mu^i Z . G(Z)$, and converge to the least fixed-point \( \Reach(X_T) = \mu Z . G(Z) \). Let \( L = \{Z_0, Z_1, \ldots, Z_n\} \subseteq X_d \times U_d \) denote the sequence recording these iteration sets, where $Z_n = \mu Z . G(Z)$. Define the index function \( j_L : X \to \mathbb{N} \cup \{\infty\} \) as
\begin{equation}
\label{eqn:j_L}
j_L(x) = \inf \{ i \in \mathbb{N} \mid x \in \pi_X(Z_i) \},
\end{equation}
with \( \inf \emptyset = \infty \). The reachability controller is given by
\begin{equation}
\label{def:construction_reach_contr}
\mathcal{C}_d(q) =
\begin{cases}
\{ v \in U_d \mid (q, v) \in Z_{j_L(q)} \}, & j_L(q) < \infty, \\
\emptyset, & \text{otherwise}.
\end{cases}
\end{equation}
For a perturbation map $\varepsilon : X \times U_d \to \mathbb{R}_{\geq 0}$, consider the perturbed symbolic model $S_d(\Sigma_\varepsilon) = (X_d, X_d^0, U_d,  \Delta_{d, \varepsilon})$ constructed in Section~\ref{Def:Sd}. Analogously, the reachability controller for $S_d(\Sigma_\varepsilon)$ is defined by
\begin{equation}
\label{def:construction_reach_contr_perturbed}
\mathcal{C}_{d,\varepsilon}(q) =
\begin{cases}
\{ v \in U_d \mid (q, v) \in Y_{j_{L_\varepsilon}(q)} \}, & j_{L_\varepsilon}(q) < \infty, \\
\emptyset, & \text{otherwise}.
\end{cases}
\end{equation}
where $L_\varepsilon = \{Y_0, Y_1, \ldots, Y_m\}$ is the analogous iterate of $S_d(\Sigma_\varepsilon)$. By construction, $\mathcal{C}_{d,\varepsilon}$ is a reachability controller for $S_d(\Sigma_\varepsilon)$ and specification~$\mathcal{H}^{X_T}_{\mathrm{Reach}}$.
Next, we provide a bound on \( \varepsilon \) ensuring the equivalence between \( \mathcal{C}_d \) and \( \mathcal{C}_{d,\varepsilon} \). 
\begin{proposition}
\label{prop:reach_equiv}
Consider the discrete-time control system \( \Sigma \) and its perturbed version \( \Sigma_\varepsilon \) in~\eqref{eqn:system}, with symbolic models \( S_d(\Sigma) \) and \( S_d(\Sigma_\varepsilon) \) defined in Section~\ref{Def:Sd}. Let \( \mathcal{C}_d \) and \( \mathcal{C}_{d,\varepsilon} \) denote the corresponding reachability controllers in~\eqref{def:construction_reach_contr} and~\eqref{def:construction_reach_contr_perturbed}, respectively, with associated sequences \( L = \{Z_0,\dots,Z_n\} \) and \( L_\varepsilon = \{Y_0,\dots,Y_m\} \). If
\(\sup_{z\in Q(x)}\varepsilon(z,v) < \gamma(Q(x), v, L), \quad \forall (x,v) \in X \times U_d,\)
then \(\forall q \in X_d, \quad \mathcal{C}_d(q) = \mathcal{C}_{d,\varepsilon}(q).
\)
\end{proposition}


In the next theorem, we derive a reachability controller for the perturbed continuous system using the synthesized controller for the disturbance-free abstraction.
\begin{theorem}
\label{thm:reach_refine}
Consider the setting of Proposition~\ref{prop:reach_equiv}. We define the controller $\mathcal{C}: X\rightrightarrows U_d$ as $\mathcal{C}(x)= C_{d}(Q(x))$. For any map \(\varepsilon: X \times U_d \rightarrow \mathbb{R}_{\geq 0}\) satisfying $\sup_{z\in Q(x)}\varepsilon(z,v) < \gamma(Q(x), v, L)$ and for all \((x, v) \in X \times U_d\), the controller \(\mathcal{C}\) is a reachability controller for the concrete perturbed system \(S(\Sigma_{\varepsilon})\) and the reachability specification $Q^{-1}(\mathcal{H}_{Reach}^{X_T})$.
\end{theorem}

In the following, we present a simplified expression of the admissible robustness margin for reachability specifications.

\begin{proposition}
\label{Prop:reach_caracterization}
Consider the discrete-time control system $\Sigma$ in~\eqref{eqn:system}, its symbolic model $S_d(\Sigma) = (X_d, X_d^0, U_d, \Delta_d)$ constructed in Section~\ref{Def:Sd}, and the reachability controller $\mathcal{C}_d$ for the target set $X_T \subseteq X_d$ defined in~\eqref{def:construction_reach_contr} with associated iterate sequence $L = \{Z_0, Z_1, \ldots, Z_n\}$, where $Z_n = \mu Z.\, G(Z)$. If $(x, v) \in  Q^{-1}(Z_n)$, then for $q=Q(x)$, we have $\gamma(q, v,L) \;=\; \eta(\;q,\;\; v,  \pi_X( Z_{j_L(q)-1}))$. 
\end{proposition}

Intuitively, for any $(q,v) \in Z^\mu_j$, all of its successors reach the lower layer $Z_{j-1}^\mu$, bringing the system closer to the target set $X_T = \pi_X(Z_1)$. As a result, computing the robustness margin based on this lower layer preserves the guaranteed progression toward $X_T$.



\section{Persistence controller}
\label{sec:6}
In this section, we consider the synthesis problem that consists in determining the robust controller such that all the behaviors of the derived system $S^\mathcal{C}$ reach the set $X_T$ and remain there for all future time and determining the maximum robustness margin for the robust controller. We first provide a formal definition of the persistence specification.
\begin{definition}
\label{Def:persistence}
Consider the system $S(\Sigma)$ in~\eqref{eqn:system},  a controller $\mathcal{C} : X \rightrightarrows U$ and its domain $\dom(\mathcal{C})$ as defined in Section~\ref{control_system}. The controller $\mathcal{C}: X\rightrightarrows U$ is a persistence controller for the target set $X_T$ if for any initial state $x_0 \in \dom(\mathcal{C})$, all maximal trajectories $\sigma \in \mathcal{H}(S^\mathcal{C}) $ are complete and satisfy $$
\exists k \in \N, \space \forall m \geq k, \quad \sigma_x(m) = x_m \in X_T.
$$

We denote the corresponding persistence specification as $\mathcal{H}_{Spec} = \mathcal{H}_{Pers}^{X_T} = \{\sigma_x  \in X^w \mid
   \exists k \in \N, \space \forall m \geq k, \quad \sigma_x(m) = x_m \in X_T\}.$

\end{definition}
The controller synthesis for the persistence problem is done using the fixed-point Algorithm~\ref{alg:persistence}. This algorithm calculates the winning set denoted by \(Z_{\infty}\), which characterizes the maximal set of state-input pairs $(x, u)$ from which the system can be steered into $X_T$ in finite time and confined there indefinitely. Additionally, the algorithm records the intermediate sets generated during each iteration leading to the computation of the final winning set. Consider the symbolic model $S_d(\Sigma)=(X_d, X_d^0, U_d, \Delta_d)$ constructed in Section~\ref{Def:Sd} and let $L^\mu = \{Z_0^\mu, Z_1^\mu, \dots, Z_n^\mu\}$, with $Z_n^\mu =Z_{\infty}$, be the outer iterations list produced by Algorithm~\ref{alg:persistence} for $S_d(\Sigma)$. These sequences provide valuable insight into the convergence behavior of the algorithm and will be used as theoretical tools in the proofs.
The winning set can written using $\mu$-calculus as $Z_{\infty}=\mu Z^\mu . \nu Z^\nu .\left(\Pre_{\Delta_d} \left(Z^\nu\right) \cap W\right) \cup \Pre_{\Delta_d}\left(Z^\mu\right), $
with $ W = X_T \times U_d$. The iterations used in the controller synthesis correspond to the outer loop iterations, satisfying the formula: $Z^\mu_{i}=\mu^i Z^\mu . \nu Z^\nu .\left(\Pre_{\Delta}\left(Z^\nu\right) \cap W\right) \cup \Pre_{\Delta}\left(Z^\mu\right),$ where $Z^\mu_0 = \emptyset \times U $.

\begin{algorithm}
\caption{Computation of $Z_{\infty}$}
\label{alg:persistence}
\begin{algorithmic}[1]
\STATE \textbf{Input:} $S(\Sigma), W = X_T \times U$ where  $ X_T \subseteq X$
\STATE \textbf{Output:} $\Pers(X_T) , L^{\nu}=\{Z^{\nu}_i \mid i \in \N \},  L^{\mu}=\{Z^{\mu}_i \mid i \in \N \}$
\STATE $Z^{\nu}_0 = X \times U$ 
\STATE $Z^{\mu}_0 :=  \emptyset \times U$
\REPEAT
    \REPEAT
        \STATE $Z^{\nu}_{j+1}:=\left(\Pre\left(Z_j^{\nu}\right) \cap W \right) \cup \Pre(Z^{\mu}_i)$
    \UNTIL  $Z^{\nu}_{j} = Z^{\nu}_{j+1}$
    \STATE $Z^{\mu}_{i+1} := Z^{\nu}_j$
\UNTIL $Z^{\mu}_{i+1} = Z^{\mu} _i$
\STATE  $Z_{\infty}:= Z_i^{\mu}$ 
\end{algorithmic}
\end{algorithm}

The persistence controllers for the symbolic model $S_d(\Sigma)$ and its perturbed version $S_d(\Sigma_\varepsilon)$ with respect to the target set $X_T$ are constructed by the same index-based rule as the reachability controllers in~\eqref{def:construction_reach_contr} and~\eqref{def:construction_reach_contr_perturbed}, with the reachability iterate sequences $L$ and $L_\varepsilon$ replaced by the outer persistence iterate sequences
$L^\mu$ and $L^\mu_\varepsilon$ produced by Algorithm~\ref{alg:persistence}.
Specifically,
\begin{equation}
\label{eq:persistence_controller}
\mathcal{C}_{d}(q) =
\begin{cases}
\{ v \in U_d \mid (q,v) \in Z^{\mu}_{j_{L^{\mu}}(q)} \}
  & \text{if } j_{L^{\mu}}(q) < \infty, \\
\emptyset & \text{otherwise,}
\end{cases}
\end{equation}
where $j_{L^\mu}(q) = \inf\{i \in \mathbb{N} \mid
q \in \pi_X(Z^\mu_i)\}$, and $\mathcal{C}_{d,\varepsilon}(q) $ is defined analogously to $\mathcal{C}_{d}(q)$, using the list $L^\mu_\varepsilon$.
By construction, $\mathcal{C}_d$ and $\mathcal{C}_{d,\varepsilon}$ are persistence controllers for $S_d(\Sigma)$ and $S_d(\Sigma_\varepsilon)$, respectively, both enforcing specification~$\mathcal{H}^{X_T}_{\mathrm{Pers}}$~\cite{tabuada2009verification,maler1995synthesis}.

The following two lemmas establish building blocks for the proof of the subsequent results, where the proofs can be found in the appendix.

\begin{lemma}
\label{lemma_inclusion_oprators}
Consider the discrete-time control systems $\Sigma$ and its symbolic model $S_d(\Sigma)=(X_d, X_d^0, U_d, \Delta_d)$ constructed in Section~\ref{Def:Sd}. Let $G, G_\varepsilon: 2^{X_d \times U_d} \to 2^{X_d \times U_d}$ be monotone operators. Define the sequences $Z_0 = Z_0^\varepsilon = X_d \times U_d,$ $\nu^{i+1}Z . G(Z_i)=  G(\nu^{i}Z . G(Z))$ and $\nu^{i+1}Z . G_\varepsilon(Z)= G_\varepsilon(\nu^iZ . G_\varepsilon(Z)).$ as in Section~\ref{Sec:fixed_point}. If $G_\varepsilon(Z) \subseteq G(Z)$ hold for all $Z \subseteq X_d \times U_d$,  then $Z_i^\varepsilon \subseteq Z_i$ for all $i \in \mathbb{N}$, in particular $\nu Z . G_\varepsilon(Z) \subseteq \nu Z . G(Z)$.
\end{lemma}

\begin{lemma}
\label{Lemma_persistence}
Consider the discrete-time control systems $\Sigma$ and its symbolic model $S_d(\Sigma)=(X_d, X_d^0, U_d, \Delta_d)$ constructed in Section~\ref{Def:Sd}. Let $L^\mu = \{Z_0^\mu, Z_1^\mu, \dots, Z_n^\mu\}$, with $Z_n^\mu =Z_{\infty}$, be the outer iterations list produced by Algorithm~\ref{alg:persistence} for $S_d(\Sigma)$ and $X_T \subseteq X_d$. We have $Z^\mu_{i+1}
= (\Pre_{\Delta_d}(Z^\mu_{i+1}) \cap W) \cup \Pre_{\Delta_d}(Z^\mu_i),$ where $W =X_T \times U$ for $i \in 1, \dots, n-1$.
\end{lemma}
In the following theorem, we provide an explicit bound on the perturbation $\varepsilon$ under which both controllers $\mathcal{C}_d$ and $\mathcal{C}_{d,\varepsilon}$ are equivalent.


\begin{proposition}
\label{prop:pers_equiv}
Consider the discrete-time control system $\Sigma$ and its perturbed version $\Sigma_\varepsilon$ in~\eqref{eqn:system} for a perturbation map $\varepsilon: X \times U_d \rightarrow \mathbb{R}_{\geq 0}$, along with their corresponding symbolic models $S_d(\Sigma) = (X_d, X_d^0, U_d, \Delta_d)$ and $S_d(\Sigma_\varepsilon) = (X_d, X_d^0, U_d,  \Delta_{d, \varepsilon})$, constructed in Section~\ref{Def:Sd}. The perturbation map is given by $\varepsilon: X \times U_d \rightarrow \mathbb{R}_{\geq 0}$. Let \( \mathcal{C}_d \) and \( \mathcal{C}_{d,\varepsilon} \) denote the corresponding persistence controllers in~\eqref{eq:persistence_controller} with associated sequences \(L^\mu = \{Z_0^\mu, Z_1^\mu, \dots, Z_n^\mu\} \) and \( L^\mu_{\varepsilon}  = \{Y_0^\mu, Y_1^\mu, \dots, Y_m^\mu\}\). If the perturbation map $\varepsilon: X \times U_d \rightarrow \mathbb{R}_{\geq 0}$ satisfies $\sup_{z\in Q(x)}\varepsilon(z,v) < \gamma(Q(x), v, L^\mu)$ for all  \((x,v) \in X \times U_d\) then $\forall q \in X_d, \quad \mathcal{C}_d(q) = \mathcal{C}_{d, \varepsilon}(q).$
\end{proposition}

The next theorem is the main result of this section. We derive a persistence controller for the perturbed continuous system using the controller of the nominal symbolic model.

\begin{theorem}
\label{thm:pers_refine}
Consider the setting of Proposition~\ref{prop:pers_equiv}.We define the controller $\mathcal{C}: X\rightrightarrows U_d$ for the continuous-space system as $\mathcal{C}(x)= C_{d}(Q(x))$. For any map \(\varepsilon: X \times U_d \rightarrow \mathbb{R}_{\geq 0}\) that satisfies $\sup_{z\in Q(x)}\varepsilon(z,v) < \gamma(Q(x), v, L^\mu)$ for all \((x, v) \in X \times U_d\), the controller \(\mathcal{C}\) is a persistence controller for the concrete perturbed system \(S(\Sigma_{\varepsilon})\) and the persistence specification $Q^{-1}(\mathcal{H}_{Pers}^{X_T})$.
\end{theorem}


In the following, we present a simplified expression of the admissible robustness margin for persistence specifications. 
\begin{proposition}
\label{Prop:pers_caracterization}
Consider the setting of Proposition~\ref{prop:pers_equiv}. For $(x, v) \in Q^{-1}(Z^\mu_n)$ with $q= Q(x)$ and $j = j_{L^\mu}(q)$, we have
\[
I\!\left(q, v, L^\mu\right)
=
\begin{cases}
Z^\mu_{j_{}}, & \text{if }  \,j_{} > 1 \text{ and } \\
&\hspace{1cm} (q,v) \notin \Pre_{\Delta_d}(Z^\mu_{j-1}) , \\
Z^\mu_{j_{}-1}, &\text{if }  \, j_{} > 1 \text{ and }\\
& \hspace{1cm } (q,v) \in \Pre_{\Delta_d}(Z^\mu_{j-1}), \\
Z^\mu_{1}, &\text{if } \, j_{} = 1.
\end{cases}
\]
which simplifies the expression of the robustness margin $\gamma(q, v,  L^\mu) = \eta(q, v, I(q,v,  L^\mu))$ for each case.
\end{proposition}

\begin{remark}
Recall that the sets $Z_i^\mu$ are ordered by inclusion. At a given pair $(q,v)$, three situations may occur. If $(q,v)$ belongs to $\Pre_{\Delta_d}(Z_{j-1}^\mu)$, then all its successors can be steered into a lower layer $Z_{j-1}^\mu$, meaning that the controller has the option to ``descend'' closer to the target set $X_T$. Thus, we calculate the robustness margin based on this lower layer. Otherwise, if $(q,v) \notin \Pre_{\Delta_d}(Z_{j-1}^\mu)$, such a descent is no longer possible, and the system is constrained to remain within the current layer $Z_j^\mu$, and we have $(q,v) \in \Pre_{\Delta_d}(Z^\mu_j) \cap W \subseteq W$
which means the successor states of $(q,v)$ reach the target $X_T$. Finally, if $(q,v) \in Z^\mu_{1}$, then $Z^\mu_{1} = Z^\mu_{1} \cap W \subseteq W$ is the maximal fixed-point contained in $W$. Consequently, any trajectory inside this set satisfies the safety specification with respect to $X_T$, and thus remains in $X_T$ for all future times. 
\end{remark}



\section{Recurrence controller}
\label{sec:7}
Let $S(\Sigma)=\left(X, X_0, U, \Delta\right)$ be a transition system in~\eqref{eqn:system} and let $X_T \subseteq X$ be a target set. We consider the synthesis problem that consists in determining a controller $\mathcal{C}$, such that all the behaviors of the corresponding control system $S^\mathcal{C}(\Sigma)$ can reach the set $X_T$ infinitely many times.
\begin{definition}
\label{Def:recurrence}
Consider the system $S(\Sigma)$ in~\eqref{eqn:system},  a controller $\mathcal{C} : X \rightrightarrows U$ and its domain $\dom(\mathcal{C})$ as defined in Section~\ref{control_system}. A controller $\mathcal{C} : X \rightrightarrows U$ for the system $S(\Sigma)$ is a recurrence controller for target set $X_T\subseteq X$ if for any initial state $x_0 \in \dom(\mathcal{C})$, all maximal trajectories $\sigma \in \mathcal{H}(S^\mathcal{C}) $ are complete and satisfy:
$$
\forall k \in \N, \quad \exists m \geq k, \quad \sigma_x(m) = x_m \in X_T.
$$
We denote the corresponding recurrence specification as $ \mathcal{H}_{Spec} = \mathcal{H}_{Rec}^{X_T} = \{\sigma_x  \in  X^w \mid 
   \forall k \in \N, \quad \exists m \geq k, \quad \sigma_x(m) = x_m \in X_T\}.$

\end{definition}
The solution to the recurrence problem is obtained via the fixed-point Algorithm~\ref{alg:recurrence}, which computes the set $Z_{\infty}= \Rec(X_T)$.  Let $L^\nu = \{Z_0^\nu, Z_1^\nu, \dots, Z_n^\nu\}$ and $L^{\mu,i} = \{Z_0^{\mu,i}, Z_1^{\mu,i}, \dots, Z_{m(i)}^{\mu,i}\}$, for $i = 0, \dots, n$, be the outer and inner iteration lists produced by Algorithm~\ref{alg:recurrence} for the symbolic model $S_d(\Sigma)$ constructed in Section~\ref{Def:Sd} and target set $X_T \subseteq X_d$, where $m(i)$ is the minimal integer satisfying $ Z_{m(i)}^{\mu,i} = Z_{m(i)+1}^{\mu,i}$ and $Z_n^\nu = Z_{\infty}$. Define the combined list
\begin{equation}
\label{combination_list}
L^\mu = \bigcup_{i=0}^{n} \{ Z_0^{\mu,i}, Z_1^{\mu,i}, \dots, Z_{m(i)}^{\mu,i} \}.
\end{equation}
The final set computed by the algorithm admits the nested fixed-point characterization
$Z_{\infty} \quad = \nu Z^\nu .\,  \mu Z^\mu.  ( \Pre_{\Delta_d}(Z^\nu)    \cap W ) \cup \Pre_{\Delta_d}(Z^\mu),$
where \(W = X_T \times U_d\). Consider the list  \(  L^{\mu,n} =  \{ Z_0^{\mu,n}, Z_1^{\mu,n}, \dots, Z_{m(n)}^{\mu,n} \}\) the sequence corresponding to the last inner iterations of Algorithm~\ref{alg:recurrence}. This list will be used for the definition of the recurrence controller, and to simplify the notation, we note it as \begin{equation}
\label{recurr_steps}
L=   \{ Z_0, Z_1, \dots, Z_{m(n)} \}
\end{equation}.

\begin{algorithm}
\caption{Computation  of $\Rec(X_T)$}
\label{alg:recurrence}
\begin{algorithmic}[1]
\STATE \textbf{Input:} $S(\Sigma)$, $W = X_T \times U$ where  $ X_T \subseteq X$
\STATE \textbf{Output:} $\Rec(X_T), L^{\nu}=\{Z^{\nu}_i \mid i = 0,\dots,n\},  L^{\mu}=\{\{Z^{\mu,i}_j \mid j = 0,\dots,m(i)\} \mid i = 0,\dots,n\}$
\STATE $Z^{\nu}_0 \gets X \times U$ 
\STATE $i \gets 0$
\REPEAT
    \STATE $Z^{\mu,i}_0 \gets \emptyset \times U$
    \STATE $j \gets 0$
    \REPEAT
        \STATE $Z^{\mu,i}_{j+1} \gets \left(\Pre_{\Delta}\left(Z^{\nu}_i\right) \cap W \right) \cup \Pre_{\Delta}(Z^{\mu,i}_j)$
        \STATE $j \gets j + 1$
   \UNTIL{$Z^{\mu,i}_{j} = Z^{\mu,i}_{j-1}$}
    \STATE $Z^{\nu}_{i+1} \gets Z^{\mu,i}_{j}$
    \STATE $i \gets i + 1$
\UNTIL{$Z^{\nu}_{i} = Z^{\nu}_{i-1}$}
\STATE $\Rec(X_T) \gets Z^{\nu}_i$
\end{algorithmic}
\end{algorithm}

The recurrence controller for $S_d(\Sigma)$ uses the last inner iterate sequence $L$ of Algorithm~\ref{alg:recurrence}, and is defined by:
\begin{equation}
\label{eq:recurrence_controller}
\mathcal{C}_d(q) =
\begin{cases}
\{ v \in U_d \mid (q,v) \in Z_{j_L(q)} \} & \text{if } j_L(q) < \infty, \\
\emptyset & \text{otherwise},
\end{cases}
\end{equation}
where $j_{L}(q) = \inf\{i \in \mathbb{N} \mid
q \in \pi_X(Z_i)\}$.  We define $\mathcal{C}_{d,\varepsilon}(q) $ analogously to $\mathcal{C}_{d}(q)$, using the list $L_\varepsilon$ which is the final inner iterate sequence of Algorithm~\ref{alg:recurrence} applied to $S_d(\Sigma_\varepsilon)$.
By construction, $\mathcal{C}_d$ and $\mathcal{C}_{d,\varepsilon}$ are recurrence controllers for $S_d(\Sigma)$ and $S_d(\Sigma_\varepsilon)$, respectively, both enforcing specification~$\mathcal{H}^{X_T}_{\mathrm{Rec}}$~\cite{tabuada2009verification,maler1995synthesis}. 
We are now ready to derive an explicit bound on the perturbation $\varepsilon$ under which the two controllers $\mathcal{C}_d$ and $\mathcal{C}_{d,\varepsilon}$ are equivalent.
\begin{proposition}
\label{prop:recu_equiv}
Consider the discrete-time control systems $\Sigma$ and $\Sigma_\varepsilon$
in~\eqref{eqn:system}, with symbolic models $S_d(\Sigma) = (X_d, X_d^0,
U_d, \Delta_d)$ and $S_d(\Sigma_\varepsilon) = (X_d, X_d^0, U_d,
 \Delta_{d, \varepsilon})$ constructed in Section~\ref{Def:Sd}, where
$\varepsilon : X \times U_d \to \mathbb{R}_{\geq 0}$ is the perturbation
map. Let $\mathcal{C}_d$ and $\mathcal{C}_{d,\varepsilon}$ be the recurrence
controllers synthesized by Algorithm~\ref{alg:recurrence} for $S_d(\Sigma)$
and $S_d(\Sigma_\varepsilon)$, respectively, with respect to the target set
$X_T \subseteq X_d$. Denote by $L^\nu = \{Z_0^\nu, \ldots, Z_n^\nu\}$ and
$L^{\mu,i} = \{Z_0^{\mu,i}, \ldots, Z_{m(i)}^{\mu,i}\}$, $i = 0, \ldots,
n$, the outer and inner iterate lists produced by Algorithm~\ref{alg:recurrence}
for $S_d(\Sigma)$, and let $L^\mu$ and $L$ be as defined
in~\eqref{combination_list} and~\eqref{recurr_steps}. The corresponding
lists for $S_d(\Sigma_\varepsilon)$ are denoted $L^\nu_\varepsilon =
\{Y_0^\nu, \ldots, Y_{n_1}^\nu\}$, $L^{\mu,i}_\varepsilon = \{Y_0^{\mu,i},
\ldots, Y_{m_\varepsilon(i)}^{\mu,i}\}$, $i = 0, \ldots, n_1$, where
$m_\varepsilon(i) \in \mathbb{N}$ is the smallest integer such that
$Y_{m_\varepsilon(i)}^{\mu,i} = Y_{m_\varepsilon(i)+1}^{\mu,i}$, and
$L^\mu_\varepsilon$, $L_\varepsilon$ are defined analogously. If $\sup_{z\in Q(x)}\varepsilon(z,v) < \gamma(Q(x), v, L^\mu), \quad \forall\, (x,v) \in X \times U_d,$ then the following holds: \(
\forall q \in X_d, \quad \mathcal{C}_d(q) = \mathcal{C}_{d, \varepsilon}(q).
\)
\end{proposition}

Note that for calculating the robustness margin we use the list $L^\mu$ whereas for the controller definition we use the list $L$. The next theorem is the main result of this section, we derive a recurrence controller for the concrete system $S(\Sigma_{\varepsilon})$.
\begin{theorem}
\label{thm:recu_refine}
Consider the setting of Proposition
~\ref{prop:recu_equiv}. Let $\mathcal C_d$ be the recurrence controller in \eqref{eq:recurrence_controller}, and define $\mathcal{C}(x)= C_{d}(Q(x))$. For any map \(\varepsilon: X \times U_d \rightarrow \mathbb{R}_{\geq 0}\) that satisfies {$\sup_{z\in Q(x)}\varepsilon(z,v) < \gamma(Q(x), v, L^\mu)$} for all \((x, v) \in X \times U_d\), the controller \(\mathcal{C}\) is a recurrence controller for the concrete perturbed system \(S(\Sigma_{\varepsilon})\) and the recurrence specification $Q^{-1}(\mathcal{H}_{Rec}^{X_T})$.
\end{theorem}

\section{Comparison of Robustness Margins}
\label{sec:9}
{In this section, }we show within the controller’s domain, the robustness margin for the four specifications studied in the previous sections is larger than the inherent robustness margin proposed in~\cite{aitsi2025symbolic}. Figure~\ref{fig:2} illustrates that the difference between the two robustness measures in the case of safety can be substantial, thus highlighting the added value and effectiveness of the proposed approach.

\begin{proposition}
\label{prop:unified_comparison}
Let $\Phi \in \{\Safe, \Reach, \Pers, \Rec\}$ be one of the four
specifications of Sections~\ref{sec:4}--\ref{sec:7}, with target set
$X_T \subseteq X_d$, and let $\mathcal{C}_d$ be the corresponding controller
synthesized for the symbolic model $S_d(\Sigma)$. Denote by $Z_\infty$ the
winning set returned by the associated fixed-point algorithm, so that
$\pi_X(Z_\infty) = \dom(\mathcal{C}_d)$, and by $L_\Phi$ the iterate list of the fixed-point algorithm governing the robustness margin of $\mathcal{C}_d$. Then, for all $(x,v) \in Q^{-1}(Z_\infty)$,
\begin{equation}
\label{eqn:unified_comparison}
\eta\big(Q(x), v, \Delta_d(Q(x), v)\big)
\;\le\;
\gamma\big(Q(x), v, L_\Phi\big).
\end{equation}
\end{proposition}

\begin{figure}[htbp]
    \centering
    \includegraphics[width=0.5\textwidth]{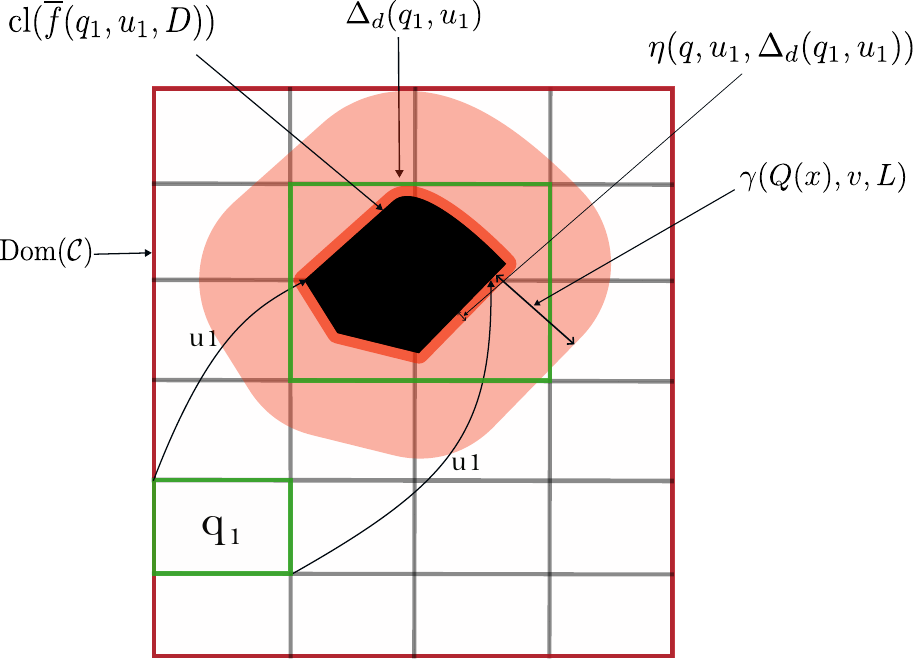}
\caption{Illustration of Proposition~\ref{prop:unified_comparison} where we compare two robustness margin for the overapproximation of the reachable set $\cl(\overline{f}(q_1, u_1, D)$, represented by the dark region, for the state \( q_1 \) under the input $u_1$. We can clearly notice the difference between the abstraction-based margin $\eta(q_1,u_1,\Delta_d(q_1,u_1))$ (dark orange) as in~\cite{aitsi2025symbolic} and the specification-based margin for the safety specification $ \gamma(q_1, u_1,L)$ (light orange) for the safety controller $\mathcal{C}$ with its domain of definition $\dom(\mathcal{C})$.}
    \label{fig:2}
\end{figure}
\section{Numerical example}
\label{sec:10}

In this section, we present numerical examples to illustrate the theoretical part of our work. We consider two examples, a first example of a double integrator with a safety specification and a mobile robot example with an LTL specification.
\subsection{Double integrator}
\label{sec:11}

We study a double-integrator system governed by
\[
\left\{\begin{array}{l}
x^1_{k+1} = x^1_k + \tau x^2_k + \dfrac{\tau^2}{2} u_k + \dfrac{\tau^2}{2} w^1_k, \\
x^2_{k+1} = x^2_k + \tau u_k + \tau w^2_k,
\end{array}\right.
\]
where the state vector \(x=(x^1,x^2)\in X=[0,6]\times[-6,6]\) consists of the position and velocity components. 
The control input \(u\in[-1,1]\), while the disturbance vector 
\(w=(w^1,w^2)\in[-0.01,0.01]\times[-0.01,0.01]\) represents additive perturbations. 
The sampling period is \(\tau=0.5\) seconds. We borrowed the same configuration as in \cite{aitsi2025symbolic} to be able to compare and show that our new results allow larger robustness margins. A symbolic abstraction of the system is constructed using the procedure described in Section~\ref{Def:Sd}. To this end, the state space \(X\) is uniformly partitioned into \(N_{x_1}=80\) intervals along the first coordinate and \(N_{x_2}=160\) intervals along the second coordinate. 
Similarly, the control input set is discretized into \(N_{u_1}=5\) as in \cite{aitsi2025symbolic}. 
The transition relation of the resulting symbolic model is then obtained via reachability analysis. The exact reachable sets can be computed by evaluating the dynamics at the extreme values of the state and disturbance variables 
(see, e.g., \cite{althoff2021set}) because we are dealing with a linear dynamics. Next, the safety fixed-point Algorithm~\ref{alg:safe} is applied to synthesize the maximal safety controller \(\mathcal{C}\) for the symbolic model \(S_d(\Sigma)\) with respect to the safe set \(X_T = [0,6]\times[-6,6]\), which yields the list of sets generated at each iteration $L$.  Figure~\ref{fig:input_heatmap} shows the deterministic controller \(\overline{\mathcal{C}}(x)=\argmax_{u \in \mathcal{C}(x)}\gamma(Q(x),u,L)\), which assigns to each state $x$ the control input that maximizes the robustness margin. Based on these controller domain, the admissible robustness margin \(\gamma: X_d \times U_d \times 2^{X_d \times U_d} \rightarrow \mathbb{R}_{\geq 0}\) is computed according to Proposition~\ref{prop:Safety_explicit}. In this example, the admissible robustness margin lies in the interval \([0.0,\, 1.5]\). Table~\ref{tab:my_table} reports the maximum admissible robustness margin over all state and input obtained for various choices of state and input discretization parameters using the notion of robustness margin in \cite{aitsi2025symbolic} and the robustness margin defined in our paper. We have two main observations. First, refining the discretization reduces the free maximum robustness margin in \cite{aitsi2025symbolic}, highlighting the inherent trade-off between abstraction accuracy and robustness. In contrast, our robustness margin exhibits the opposite trend. This trade-off is no longer limiting, and finer abstractions yield improved robustness. Second, our robustness margin is $7$ times higher than the robustness margin in~\cite{aitsi2025symbolic} in the worst case and $48$ times higher in the optimal case.
\begin{figure}[htbp]
        \centering
        \includegraphics[scale=0.48]{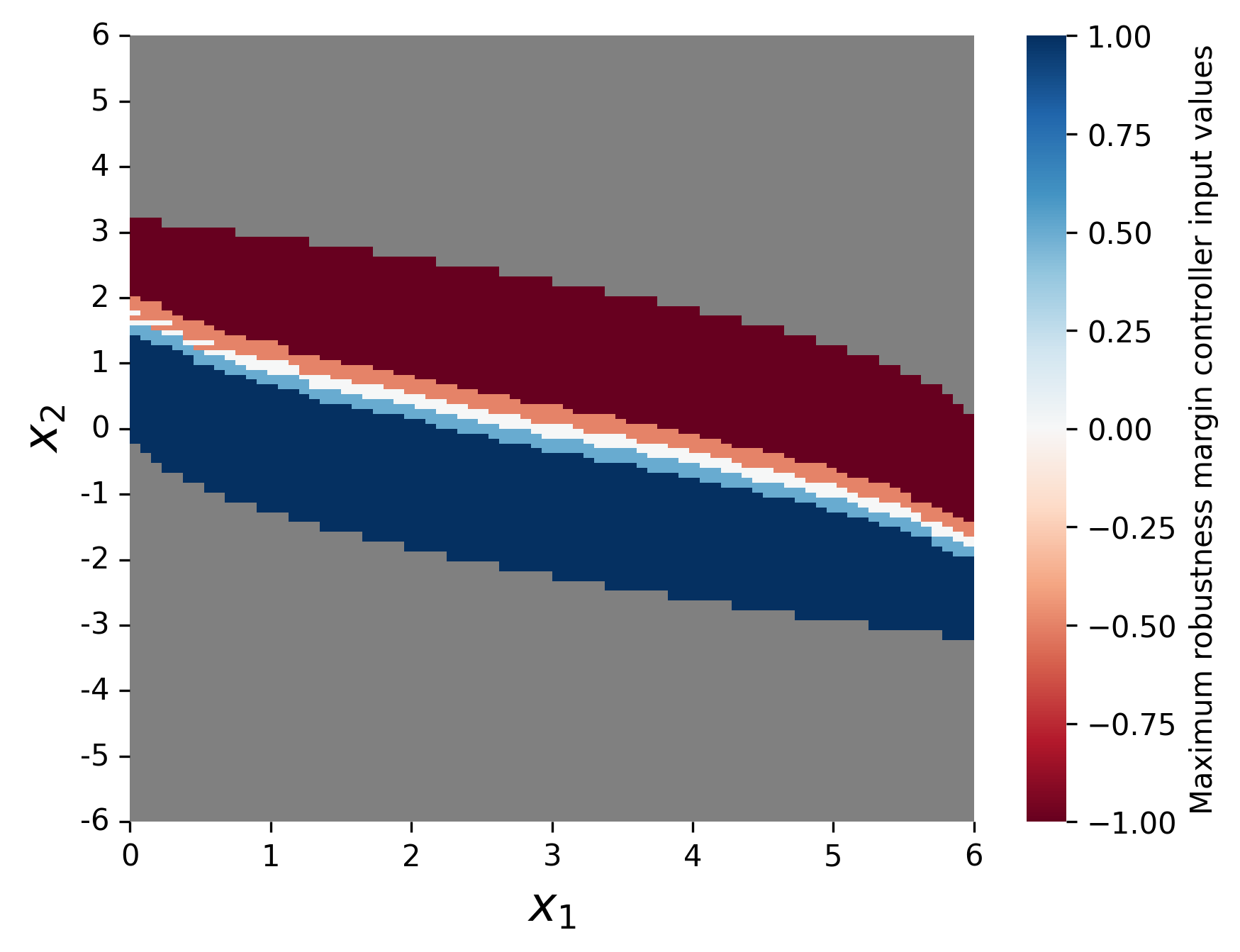}
        \caption{The input value maximizing the robustness margin for each state.}
        \label{fig:input_heatmap}
\end{figure}
The heatmap in the right of Figure~\ref{fig:two_heatmap} illustrates the corresponding maximal admissible robustness margin $ \gamma(Q(x),\overline{\mathcal{C}}(x),\mathrm{dom}(\overline{\mathcal{C}}))$ for each state \(x\in\mathrm{dom}(\overline{\mathcal{C}})\). As expected, the robustness values increase as we move farther from the border of the safety domain. The heatmap at the left shows the maximal values using the robustness in~\cite{aitsi2025symbolic}. These values are significantly smaller (dark red) compared to our robustness margins that can be seen in the dark blue on the heatmap at the right.
\begin{figure*}[htbp]
        \centering
        \includegraphics[width=13cm, height=5cm]{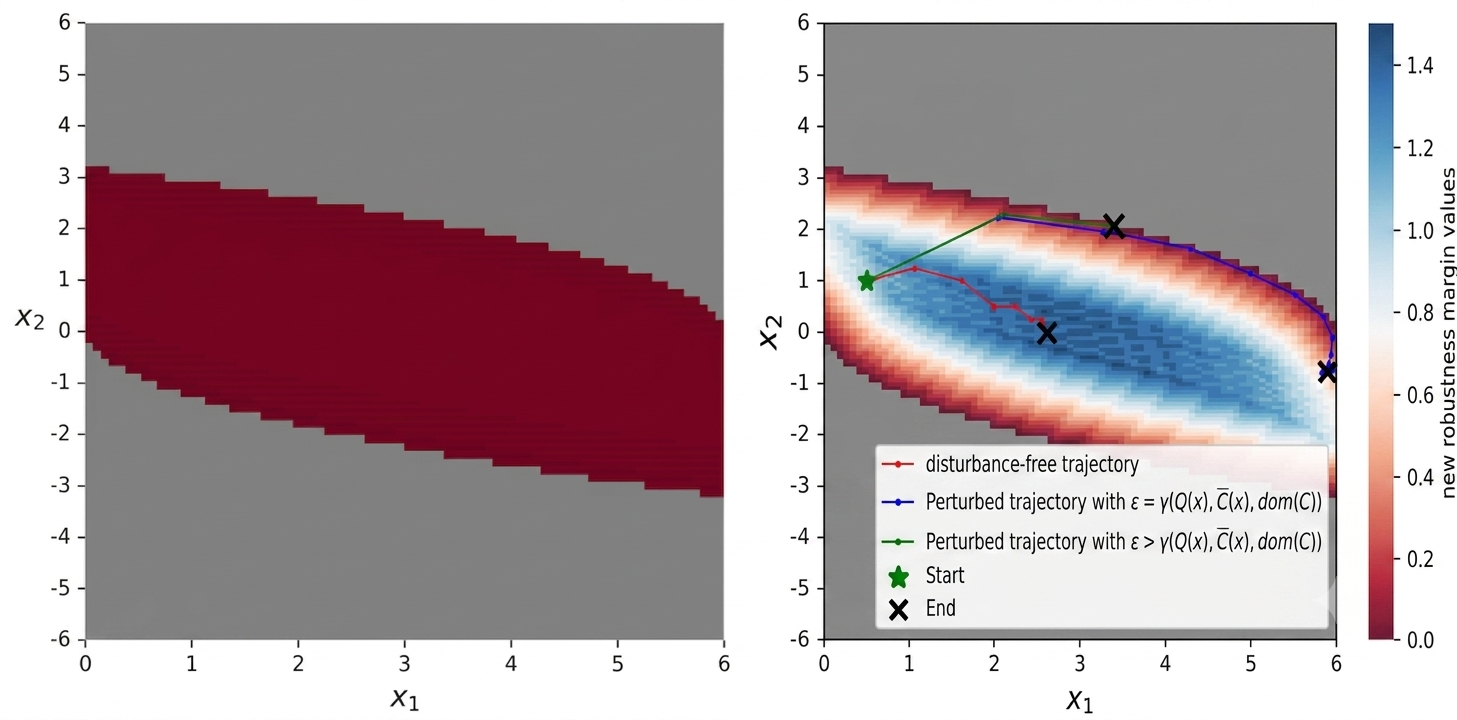}
        \caption{Two heatmaps with the same graded colorbar. The left heatmap shows the maximal free admissible robustness margin (computed as in \cite{aitsi2025symbolic}) for each state. The right heatmap shows the admissible robustness margin for the maximal controller described in Section~\ref{sec:11} along with trajectories under different perturbations.}
        \label{fig:two_heatmap}
 \end{figure*}

Finally, the right panel of Figure~\ref{fig:two_heatmap} illustrates closed-loop trajectories starting from the initial condition $x_0=(0.5,1)$. The red curve corresponds to the disturbance-free closed-loop system, while the blue curve represents the trajectory of the system under maximal additive disturbances.
In these cases, the perturbed trajectory remains within the safe set, in agreement with Theorem~\ref{thm:safety_refine}. The green curve illustrates the evolution of the disturbed system when the disturbance exceeds the maximal robustness margins that go beyond the controller domain after two time steps illustrating the tightness of our results.


\begin{table}[htbp]
\centering
\scriptsize
\setlength{\tabcolsep}{4pt}
\renewcommand{\arraystretch}{0.95}
\begin{tabular}{@{} c l cccc @{}}
\toprule
$N_u$ & Source & $N_x{=}1600$ & $N_x{=}3200$ & $N_x{=}6400$ & $N_x{=}12800$ \\
\midrule
\multirow{3}{*}{3}
 & Paper \cite{aitsi2025symbolic} & 0.118  & 0.044 & 0.025 & 0.031 \\
 & This paper               & 0.899  & 0.975 & 1.35  & 1.5   \\
 &                          & ${\times}7\uparrow$ & ${\times}22\uparrow$ & ${\times}54\uparrow$ & ${\times}48\uparrow$ \\
\addlinespace[3pt]
\multirow{3}{*}{5}
 & Paper \cite{aitsi2025symbolic} & 0.118  & 0.056 & 0.056 & 0.031 \\
 & This paper               & 0.899  & 0.975 & 1.35  & 1.5   \\
 &                          & ${\times}7\uparrow$ & ${\times}17\uparrow$ & ${\times}24\uparrow$ & ${\times}48\uparrow$ \\
\addlinespace[3pt]
\multirow{3}{*}{10}
 & Paper \cite{aitsi2025symbolic} & 0.129  & 0.055 & 0.055 & 0.047 \\
 & This paper               & 0.899  & 0.975 & 1.35  & 1.5   \\
 &                          & ${\times}6\uparrow$ & ${\times}17\uparrow$ & ${\times}24\uparrow$ & ${\times}31\uparrow$ \\
\bottomrule
\end{tabular}
\caption{Maximum robustness over all states and inputs, computed using the
robustness notion from~\cite{aitsi2025symbolic} and that of
Proposition~\ref{prop:Safety_explicit}, for the safety specification under
different discretizations $N_x = N_{x_1} \times N_{x_2}$ (the total number of discrete states) and $N_u$ (the number of symbolic inputs).}
\label{tab:my_table}
\end{table}


\subsection{Mobile robot}
\label{sec:12}
In this example, we study a unicycle-type mobile robot. The system state consists of the planar position
$\left(x^1, x^2\right) \in [0,10] \times [0,10]$ and the orientation angle $x^3 \in [-\pi,\pi]$.
The robot evolves according to
\begin{equation}
\label{diff_drive}
\left\{
\begin{aligned}
x^1_{k+1} &= x^1_k + \tau\big(u^1_k \cos(x^3_k) + w^1_k\big), \\
x^2_{k+1} &= x^2_k + \tau\big(u^1_k \sin(x^3_k) + w^2_k\big), \\
x^3_{k+1} &= x^3_k + \tau\big(u^2_k + w^3_k\big) \; (\mathrm{mod}\; 2\pi),
\end{aligned}
\right.
\end{equation}
where the control input $(u^1, u^2) \in [0.25,1] \times [-1,1]$ corresponds to the linear and angular velocities, respectively. The terms $w^i_k \in [-0.05,0.05]$, for $i=1,2,3$, model bounded disturbances, and the sampling period is set to $\tau = 1$. To formalize the specification, we first introduce five disjoint regions within the domain \( X = [0, 10] \times [0, 10] \times [-\pi, \pi] \). These regions: center $C_1$, center $C_2$, center $C_3$, Obstacle $R_1$ and obstacle $R_2$, are shown in Figure~\ref{fig:combined}, and their numerical values are given in Table~\ref{fig:regions}. To construct the symbolic model, the state space \( {X} \) is uniformly partitioned using a Cartesian grid with 
\( N_{x_1} = 120 \), \( N_{x_2} = 100 \), and \( N_{x_3} = 60 \) subdivisions along the respective dimensions as described in Section~\ref{Def:Sd}. The control inputs are uniformly discretized into \( N_{u_1} = 3 \) and \( N_{u_2} = 5 \) values. The transition relation of the symbolic model is computed via reachability analysis. Specifically, the reachable sets are over-approximated using the growth-bounds-based approach in~\cite{ReissigWeberRungger17}. We consider a combined specification formed through the sequential composition of three basic specifications: two persistence specifications followed by one recurrence specification, with the additional requirement that the system avoids obstacle regions \(R_1\) and \(R_2\) at all times. We first synthesis a safety specification over $X \backslash \{R_1 , R_2\}$ to remove the non safe states. Beginning from an initial state \(x_0\), the system must satisfy persistence within region \(C_2\) for $100$ time steps. The final state of this phase then serves as the initial condition for the second phase, which requires persistence in region \(C_1\) for another $100$ steps. The resulting endpoint state initiates the third and final phase, which requires the system to satisfy the recurrence for region \(C_3\) over $100$ steps. For each specification, we construct the nondeterministic controllers \(\mathcal{C}_1\), \(\mathcal{C}_2\), and \(\mathcal{C}_3\) using Theorems~\ref{thm:pers_refine} and~\ref{thm:recu_refine}, with their corresponding lists $L_1, L_2$ and $L_3$. Subsequently, we deduce the maximal robustness controllers defined by \(\overline{\mathcal{C}}_i(x) =\argmax_{u \in \mathcal{C}_i(x)}\gamma(Q(x),u,L_i)\). The margins are computed leveraging the robustness margin characterizations as given in Propositions~\ref{Prop:pers_caracterization} for the persistence specification.
\begin{table}[h!]
\centering
\renewcommand{\arraystretch}{1.5}
\begin{tabular}{|c|c|}
\hline
\textbf{Region} & \textbf{Coordinates} \((x^1, x^2, x^3)\) \\
\hline
Center $C_1$& \([0.3, 3] \times [8, 9.8] \times [-\pi, \pi]\) \\
Center $C_2$ & \([4, 6.5] \times [0.2, 2] \times [-\pi, \pi]\) \\
Center $C_3$ & \([7.3, 9.9] \times [8, 9.8] \times [-\pi, \pi]\) \\
Obstacle $R_1$ & \([2, 3] \times [0, 7] \times [-\pi, \pi]\) \\
Obstacle $R_2$ & \([6, 7] \times [3, 10] \times [-\pi, \pi]\) \\
\hline
\end{tabular}
\caption{Regions defining the specification.}
\label{fig:regions}
\end{table}

Figure~\ref{fig:combined} shows the combined specification for the initial state $x_0 = (1, 2)$ and displays two distinct closed-loop trajectories starting from the same initial condition. 
\begin{figure}[htbp]
    \centering
    \includegraphics[width=0.4\textwidth]{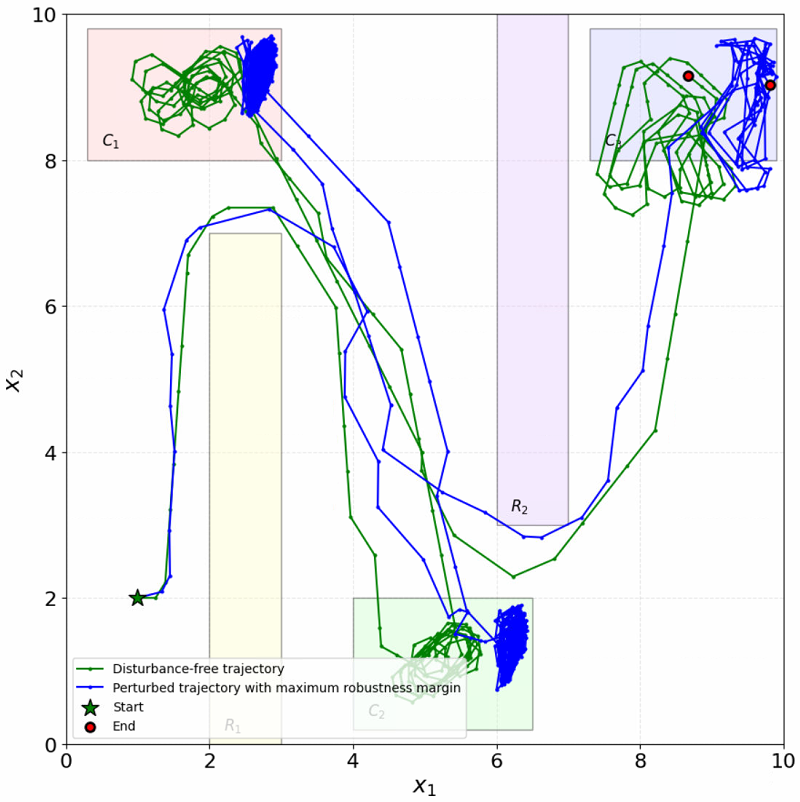}
\caption{Two controlled trajectories of the system in~\eqref{diff_drive} starting from $x(0)=(1, 2)$ over $300$ time steps, both satisfying the specification described in mobile robot Section~\ref{sec:12}. Green: evolution without disturbances ($\varepsilon = 0$). Blue: evolution with disturbance at the maximum robustness margin.}
    \label{fig:combined}
\end{figure}
The green trajectory represents the nominal case, where the system evolves without any external disturbance. The blue trajectory corresponds to the robust case, where at each step the system is subjected to the maximal admissible disturbance, calculated respectively via Proposition~\ref{Prop:pers_caracterization} for the persistence phases, implemented as the perturbation \(\gamma_i(Q(x), \overline{\mathcal{C}}_i(x), L_i)\). One can see that the control objective is achieved, which is consistent with Theorem~\ref{thm:pers_refine} and Theorem~\ref{thm:recu_refine}. Two video simulations of the robot satisfying the specification for the nominal system, implemented using PyBullet~\cite{pybullet}, are available at the following links: \href{https://youtu.be/mNZ_SWpXOAQ}{https://youtu.be/h9OsOwvBguk} and \href{https://youtu.be/8zZ4-2YEVyw}{https://youtu.be/4O9C\_uTQX7g}. The first video shows the nominal case without perturbations, while the second illustrates the scenario with disturbances corresponding to external slippery road condition.


\begin{figure}[htbp]
    \centering
    \includegraphics[width=0.4\textwidth]{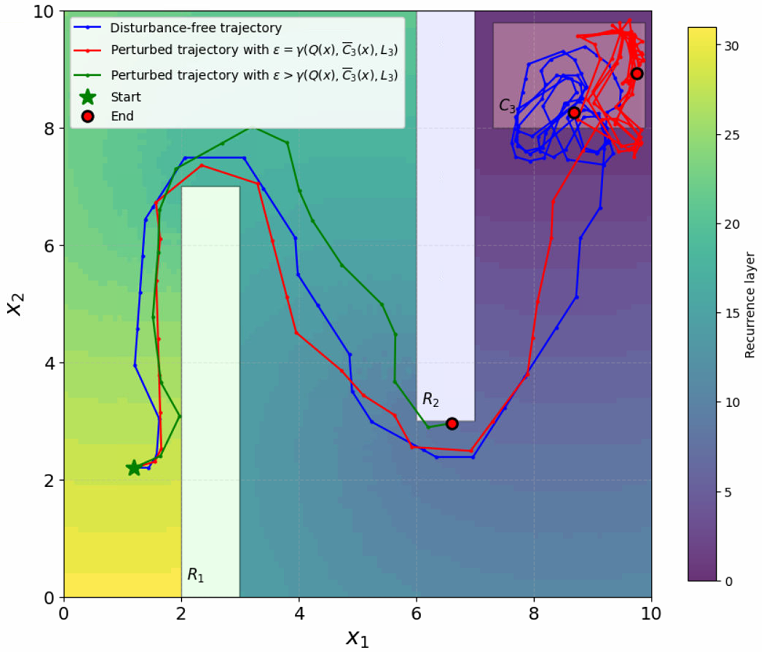}
\caption{Three trajectories of the system in~\eqref{diff_drive} starting from $x_0 = (1.2, 2.2)$ over $100$ time steps under
different perturbations.}
    \label{fig:recurrence}
\end{figure}
In Figure~\ref{fig:recurrence}, we focus only on the recurrence specification over the center $C_3$. The graded color layers represent the sets \(L_3 = \{Z_0, Z_2, \dots, Z_{32}\}\), where each color corresponds to the region \(Z_{i+1} \backslash Z_{i}\). The darkest violet denotes set \(Z_0\), which must be visited infinitely often, while the lightest yellow corresponds to the layer \(Z_{32} \backslash Z_{31}\). Here $Z_{32}$ is the fixed-point of the recurrence algorithm where $Z_i \in Z_{32}$ for $i = 0,1 \dots 31.$ We simulate three closed-loop trajectories from the initial state \(x_0=(1.2, 2.2)\) under the maximal robustness controller \(\overline{\mathcal{C}}_3\). The first (blue) trajectory corresponds to nominal control without disturbances. The second (red) trajectory shows the system under maximal robustness perturbation given by \(\gamma(Q(x), \overline{\mathcal{C}}_3(x), L_3)\) for each state $x \in X$. Both trajectories satisfy the specification: they avoid obstacles \(R_1\) and \(R_2\), reach the center \(C_3\), and revisit it infinitely often (simulated up to $100$ steps). Additionally, each state \(x_k \in Z_i\) transitions to the next layer satisfying \(x_{k+1} \in Z_{i-1}\). In contrast, the green trajectory illustrates a case where the applied disturbance exceeds the computed maximal robustness margin. This violation causes the system to fail the specification by entering in obstacle $R_2$.  The sharp transition from satisfaction to violation at this threshold demonstrates the tightness of our result. The admissible robustness margin lies in the range \([0.0, 0.7]\), with a mean robustness margin in each direction equal to \(0.0794\). The computation times are: abstraction construction ($2$\,min\,$2$\,s), recurrence fixed-point iteration ($2$\,min\,$4.68$\,s), controller synthesis ($1$\,min\,$18.37$\,s) using a processor $2.20$GHz Intel(R) Xeon(R) Platinum $56$ cores CPU.

\section{Conclusion}
We introduce a new concept of specification-based maximal robustness margin for discrete-time dynamical systems. This quantity characterizes the largest admissible disturbance under which the system still satisfies a given temporal specification. The key innovation is the specification-dependent robustness margin that ensures fixed-point computations remain invariant under perturbations. By computing margins that preserve the predecessor operator used in the fixed-point synthesis algorithm, we derive non-uniform, state- and input-dependent robustness bounds tailored to safety, reachability, persistence, and recurrence. Theoretical guarantees show our margins strictly exceed previously known bounds, reducing conservatism. Our method enables the synthesis of controllers that are robust to disturbances below the characterized margin, guaranteeing specification satisfaction for both the nominal and perturbed systems. Simulation results on two numerical examples confirm the practical effectiveness of the approach. Future work will focus on extending these results to stochastic systems with stochastic uncertainties.

\bibliographystyle{IEEEtran}
\bibliography{sources}

\appendix
We have the following auxiliary result for which the proof is straightforward, and thus omitted.
\begin{lemma}
\label{lem:2}
Consider the discrete-time control system $\Sigma$ in (\ref{eqn:system}) and its associated symbolic model $S_d(\Sigma)=(X_d, X_d^0, U_d, \Delta_d)$ constructed in Section~\ref{Def:Sd}. Let $(q,v) \in X_d \times U_d$ and $A \subseteq B \subseteq  X_d$. We have that $\eta(q,v,A) \leq \eta(q,v,B)$.
\end{lemma}

\subsection{Proof of Lemma~\ref{Proposition:behavior}}
\begin{proof}
To show the result, we show that for all \(i \in \N\) we have $u_i \in  \mathcal{C}_{d, \varepsilon}(q_i)$ and $q_{i+1} \in \Delta_{d, \varepsilon}(q_i, u_i)$.
Let \(i \in \N\), first we have \(u_i \in \mathcal{C}(x_i) = \mathcal{C}_{d, \varepsilon}(Q(x_i)) = \mathcal{C}_{d, \varepsilon}(q_i)\) therefore the first condition is satisfied.
For the second condition, we need to show that \(q_{i+1} \in  \Delta_{d, \varepsilon}(q_i, u_i)\) which is equivalent to proving that \(\{\cl(\overline{f}(q_i,u_i,D))+\mathcal{B}_{\varepsilon_d(q_i,u_i)}(0)\} \cap \cl(q_{i+1}) \neq \emptyset\).
Using the fact that \( \mu(x_i, u_i) \le \varepsilon(x_i, u_i) \leq \sup_{x\in q_i} \varepsilon(x,u_i) = \varepsilon_d(q_i, u_i)\) where $\varepsilon_d$ defined in~\eqref{epsilon_d}, one gets \( \cl(\overline{f}(x_i,u_i,D))+\mathcal{B}_{\mu(x_i,u_i)}(0) \subseteq \cl(\overline{f}(q_i,u_i,D))+\mathcal{B}_{\varepsilon_d(q_i,u_i)}(0)\). By definition, we have \(x_{i+1} \in f(x_i,u_i,D)+\mathcal{B}_{\mu(x_i,u_i)}(0)\), which implies \(x_{i+1} \in \cl(\overline{f}(q_i,u_i,D))+\mathcal{B}_{\varepsilon_d(q_i,u_i)}(0)\). Moreover, since \(x_{i+1} \in Q(x_{i+1}) = q_{i+1}\), it implies \(\cl(q_{i+1}) \cap \{ \cl(\overline{f}(q_i,u_i,D))+\mathcal{B}_{\varepsilon_d(q_i,u_i)}(0)\}   \neq \emptyset\). Therefore, \(q_{i+1} \in  \Delta_{d, \varepsilon}(q_i, u_i)\), which concludes the proof.
\end{proof}

\subsection{Proof of Proposition~\ref{Prop:refinement}}

\begin{proof}
Let \(\sigma = ((x_0,u_0), (x_1,u_1), \ldots) \in \mathcal{H}(S^\mathcal{C}(\Sigma_{\mu}))\) where \(x_0 \in \dom(\mathcal{C})\). Since $\mathcal{C}(x)=\mathcal{C}_d(Q(x))=\mathcal{C}_{d,\varepsilon}(Q(x))$ for all $x\in X$, one gets from Lemma~\ref{Proposition:behavior} the existence of a behavior \(\sigma_d = ((q_0,u_0), (q_1,u_1), \ldots) \in  \mathcal{H}(S^{\mathcal{C}_{d,\varepsilon}}_{d}(\Sigma_{\varepsilon}))\) with $q_i = Q(x_i)$ for all $i \in \N$.
Moreover, we know that $x_0 \in \dom(\mathcal{C})$, thus \(q_0=Q(x_0) \in \dom(\mathcal{C}_{d, \varepsilon})\). Therefore, from the fact that $\mathcal{C}_{d,\varepsilon}$ is a controller for the system $S_d(\Sigma_\varepsilon)$ and specification $\mathcal{H}_{Spec}$ it follows $\sigma_{d,x}=(q_0,q_1,q_2,\ldots) \in \mathcal{H}_{Spec}$. Hence, we deduce that $\sigma_x=(x_0,x_1,x_2,\ldots) \in Q^{-1}(\mathcal{H}_{Spec})$. Hence, \(\mathcal{C}\) is a controller for the perturbed system \(S(\Sigma_{\mu})\) and the specification $Q^{-1}(\mathcal{H}_{Spec})$.
\end{proof}

\subsection{Proof of Theorem~\ref{Thrm:Pre_mu}}

\begin{proof}
Consider $j \in \{0,1,2,\ldots,m\}$, we first prove the first inclusion \( \Pre_{\Delta_d}(Z_j) \subseteq  \Pre_{ \Delta_{d, \varepsilon}}(Z_j) \). 
To simplify the notation, we denote $X_j = \pi_X(Z_j) \subseteq X_d$.
Consider \( (q, v) \in \Pre_{\Delta_d}(Z_j) \), which implies that 
\begin{equation}
\label{eqn:inclusion}
 \Delta_d(q,v)  \subseteq X_j.
\end{equation}
By definition of $I$ in~\eqref{Def:gamma_l}, the inclusion
\eqref{eqn:inclusion} gives $(q,v) \in \Pre_{\Delta_d}(Z_j)$ with $Z_j \in L$,
hence
\begin{equation}
\label{eqn:proofthm8}
   \bigcap\limits_{i=0}^m \{\pi_X({Z}_i)  \mid Z_i \in {L}, \;(q, v) \in Pre_{\Delta_d}(Z_i) \}   \subseteq {X}_j.
\end{equation}
Using the definition of the perturbed symbolic model, we have $\varepsilon_d(q, v) = \sup_{x\in q} \varepsilon(x,v) <  \gamma(Q(x),v,L)= \gamma(q,v,{L})$, and according to the definition of the maps $\eta$ and $\gamma$ in \eqref{Def:eta_a} and \eqref{Def:gamma_l}, respectively, we deduce that $ \varepsilon_d(q, v) < \gamma(q, v, {L})  = \eta(q, v,   I(q, v, L)  ) \leq \eta(q,v, {X}_j)$. Where the last inequality comes from~\eqref{eqn:proofthm8} and Lemma~\ref{lem:2}.
Moreover, we have $\cl(\overline{f}(q,v,D)) \subseteq  \Int(Q^{-1}({X}_j))$ due to Equation~\eqref{eqn:inclusion}.
Hence, one gets from the definition of the map $\eta$ in~\eqref{Def:eta_a} and from $\varepsilon_d(q, v) < \eta(q,v, {X}_j)$ that
\begin{equation}
\label{eqn:3}
\cl(\overline{f}(q,v,D))+\mathcal{B}_{\varepsilon_d(q,v)}(0) \subseteq  \Int(Q^{-1}({X}_j)).
 \end{equation}
Now consider $q' \in X_d$ such that $\cl(q')\cap\{  \cl(\overline{f}(q,v,D))+\mathcal{B}_{\varepsilon_d(q,v)}(0)\}\neq \emptyset$. We have from (\ref{eqn:3}) that $ \{  \cl(\overline{f}(q,v,D))+\mathcal{B}_{\varepsilon_d(q,v)}(0)\} \subseteq \Int(Q^{-1}({X}_j))=\Int(\bigcup\limits_ {q_1\in X_j)}\{q_1\})$. Hence, it follows that $\cl(q') \cap \Int(\bigcup\limits_ {q_1\in X_j}\{q_1\}) \neq \emptyset$, which in turn implies that $q' \in X_j$. Therefore, \( (q, v) \in  \Pre_{ \Delta_{d, \varepsilon}}(Z_j) \). Now, we prove that \(  \Pre_{ \Delta_{d, \varepsilon}}(Z_j) \subseteq \Pre_{\Delta_d}(Z_j) \). For any \( (q, v) \in  \Pre_{ \Delta_{d, \varepsilon}}(Z_j) \), we have $  \Delta_{d, \varepsilon}(q, v) \subseteq \pi_X(Z_j)$, and since $\Delta_{d}(q, v) \subseteq \Delta_{d, \varepsilon}(q, v)  $ the result follows directly.
\end{proof}

\subsection{Proof of Theorem~\ref{Thrm_pre_mu_neq}}

\begin{proof}
Set $A^* = I(q^*, v^*, L)$, so that $\gamma(q^*, v^*, L) = \eta(q^*, v^*, A^*)$. By the definition of $\eta$, $\varepsilon_d(q^*, v^*) > \gamma(q^*, v^*, L)$  implies $\mathrm{cl}(\overline f(q^*, v^*, D)) + \mathcal B_{\varepsilon_d(q^*, v^*)}(0) \not\subseteq \mathrm{Int}(Q^{-1}(A^*)).$ Pick $y$ in the left-hand set but not in $\mathrm{Int}(Q^{-1}(A^*))$. Since $\{q : q \in X_d\}$ partitions $\mathbb R^n$, there exists $q' \in X_d$ with $y \in \mathrm{cl}(q')$, and $q'$ can be chosen outside $A^*$. The relation $y \in \mathrm{cl}(\overline f(q^*, v^*, D)) + \mathcal B_{\varepsilon_d(q^*, v^*)}(0)$ together with $y \in \mathrm{cl}(q')$ gives, by definition of $\Delta_{d,\varepsilon}$, $q' \in \Delta_{d,\varepsilon}(q^*, v^*).$ Since $q' \notin A^* = \bigcap\limits_{i=0}^m \{\pi_X({Z}_i)  \mid Z_i \in {L}, \; (q^*, v^*) \in Pre_{\Delta_d}(Z_i) \} $, there exists $j^* \in \{0, \dots, m\}$ with $(q^*, v^*) \in \mathrm{Pre}_{\Delta_d}(Z_{j^*})$ and $q' \notin \pi_X(Z_{j^*})$. Hence $\Delta_{d,\varepsilon}(q^*, v^*) \not\subseteq \pi_X(Z_{j^*})$, so $(q^*, v^*) \notin \mathrm{Pre}_{\Delta_{d,\varepsilon}}(Z_{j^*})$, which conclude the results.
\end{proof}

\subsection{Proof of Proposition~\ref{prop:safety_equiv}}

\begin{proof}
To show the result, we first show by recurrence that \( Z_i = Y_i \) for all \( i \in \mathbb{N} \). We have that \( Z_0 = Y_0 = X_d \times U_d \). Now consider \( i \in \mathbb{N} \), assume \( Z_i = Y_i \). It follows that \( Z_{i+1} = Y_{i+1} = \Pre_{\Delta_d}(Z_i) \cap W = \Pre_{ \Delta_{d, \varepsilon}}(Y_i) \cap W \), using the fact that \( \Pre_{\Delta_d}(Z_i) = \Pre_{ \Delta_{d, \varepsilon}}(Y_i) \) as per Theorem~\ref{Thrm:Pre_mu}.
By induction, we conclude that \( Z_i = Y_i \) for all \( i \in \mathbb{N} \). Thus, $m = n$ and consequently \( \pi_X(Z_n) = \pi_X(Y_m) = \dom(\mathcal{C}_d) = \dom(\mathcal{C}_{d, \varepsilon}) \). Moreover, for every \( q \in \dom(\mathcal{C}_d) \), \( \mathcal{C}_d(q) = \mathcal{C}_{d, \varepsilon}(q) = \{ v \in U \mid (q,v) \in Z_n \}, \)
otherwise \( q \notin \pi_X(Z_n) \), implies \( \mathcal{C}_d(q) = \mathcal{C}_{d, \varepsilon}(q) = \emptyset \).
\end{proof}

\subsection{Proof of theorem~\ref{thm:safety_refine}}

\begin{proof}
Consider a map \(\varepsilon: X \times U_d \rightarrow \mathbb{R}_{\geq 0}\) satisfying $\sup_{z\in Q(x)}\varepsilon(z,v) < \gamma(Q(x), v, L)$ and the symbolic model of the perturbed system \(S_{d}(\Sigma_{\varepsilon})=(X_d, X_d^0, U_d,  \Delta_{d, \varepsilon})\). From Proposition~\ref{prop:safety_equiv} where $\mathcal{C}_{d, \varepsilon}$ is a safety controller for \(S_{d}(\Sigma_{\varepsilon})\) and safe set $X_T$, we have $\mathcal{C}_d(q) = \mathcal{C}_{d, \varepsilon}(q)$. Hence, it follows from Proposition~\ref{Prop:refinement} that the controller \(\mathcal{C}\) is a controller for the perturbed system \(S(\Sigma_{\varepsilon})\) and the safety specification $Q^{-1}(\mathcal{H}_{Safety}^{X_T})$.
\end{proof}

\subsection{Proof of Proposition~\ref{prop:Safety_explicit}}
\begin{proof}
From the definition of the monotone function $G$ of the safety specification, we have $Z_n \subseteq Z_{n-1} \dots \subseteq Z_0$. Moreover, since the $\Pre_{\Delta_d}$ operator is monotone, we have $ \Pre_{\Delta_d}(Z_n) \subseteq \Pre_{\Delta_d}(Z_{n-1}) \dots \subseteq \Pre_{\Delta_d}(Z_0)$. Using the fact that $Z_n = \Pre_{\Delta_d}(Z_n) \cap W$ one gets $Z_n \subseteq \Pre_{\Delta_d}(Z_n)$. Therefore, for $(q,v) \in Z_n$, we have $\bigcap\limits_{i=1}^n  \{ Z_i \in L \mid   (q, v) \in \operatorname{Pre}_{\Delta_d}(Z_i)  \}  = Z_n$. Hence, $\gamma(Q(x), v, L) = \eta(Q(x), v, \pi_X(Z_n)) = \eta(Q(x), v,\dom(\mathcal{C}_d))$.
\end{proof}

\subsection{Proof of Proposition~\ref{prop:reach_equiv}}

\begin{proof}
We show by induction that \( Z_i = Y_i \) for all \( i \in \mathbb{N} \). The base case holds since \( Z_0 = Y_0 = \emptyset \times U_d \). Assume \( Z_i = Y_i \). Then $Z_{i+1} = \Pre_{\Delta_d}(Z_i) \cup W$ and $Y_{i+1} = \Pre_{ \Delta_{d, \varepsilon}}(Y_i) \cup W.$ By Theorem~\ref{Thrm:Pre_mu}, \( \Pre_{\Delta_d}(Z_i) = \Pre_{ \Delta_{d, \varepsilon}}(Y_i) \), hence \( Z_{i+1} = Y_{i+1} \). Thus \( Z_i = Y_i \) for all \( i \), implying \( n = m \) and \( L = L_\varepsilon \). Consequently, \( j_L(q) = j_{L_\varepsilon}(q) \) for all \( q \in X_d \), which yields $\mathcal{C}_d(q) = \mathcal{C}_{d,\varepsilon}(q).$
\end{proof}

\subsection{Proof of Proposition~\ref{Prop:reach_caracterization}}

\begin{proof}
From the definition of the monotone function $G$ of the reachability , we have $Z_0 \subseteq Z_{1} \dots \subseteq Z_n$. Moreover, since the $\Pre_{\Delta_d}$ is monotone, we have
\begin{equation}
\label{eqn:Pre_reach}
  \Pre_{\Delta_d}(Z_0) \subseteq \Pre_{\Delta_d}(Z_{1}) \dots \subseteq \Pre_{\Delta_d}(Z_n).
\end{equation}
Now consider $(x, v) \in  Q^{-1}(Z_n)$ which implies $(q, v)  = (Q(x), v) \in  Z_n$. We have that $(q, v) \in Z_{j_L(q)}= \Pre_{\Delta_d}(Z_{j_L(q)-1}) \\ \cup W$ with $W = X_T \times U_d$ and $j_L(q) \in \{2,\ldots,n\}$ as the objective is to reach $Z_1 = W$. Thus,  $(q, v) \in \Pre_{\Delta_d}(Z_{j_L(q) -1})$. By the minimality definition of $j_L(q)$ in (\ref{eqn:j_L}), for all $i< j_L(q)$, $(q, v) \not\in Z_i= \Pre_{\Delta_d}(Z_{i-1}) \cup W$, which implies $(q, v) \not\in \Pre_{\Delta_d}(Z_{i-1})$.
Hence, it follows from (\ref{eqn:Pre_reach}) for $(q, v)  \in Z_n$ that  $ \bigcap\limits_{i=0}^n  \{ Z_i \in L \mid   (q, v) \in \Pre_{\Delta_d}(Z_i)  \}  = Z_{j_L(q)-1}$. Hence, one gets that $ \gamma(q, v, L) = \eta(q, v, \pi_X(Z_{j_L(q) -1})).$ That concludes the result.
\end{proof}

\subsection{Proof of Theorem~\ref{thm:reach_refine}}

\begin{proof}
Consider a map \(\varepsilon: X \times U_d \rightarrow \mathbb{R}_{\geq 0}\) satisfying  $\sup_{z\in Q(x)}\varepsilon(z,v) < \gamma(Q(x), v, L)$ and the symbolic model of the perturbed system \(S_{d}(\Sigma_{\varepsilon})=(X_d, X_d^0, U_d, \Delta_{d,\varepsilon_d})\). Let $\mathcal{C}_{d, \varepsilon}$ the reachability controller for \(S_{d}(\Sigma_{\varepsilon})\) and the target set $X_T$, from Proposition~\ref{prop:reach_equiv} we have $\mathcal{C}_d = \mathcal{C}_{d, \varepsilon}$. Hence, it follows from Proposition \ref{Prop:refinement} that the controller \(\mathcal{C}\) is a reachability controller for the perturbed system \(S(\Sigma_{\varepsilon})\) and the reachability specification $Q^{-1}(\mathcal{H}_{Reach}^{X_T})$.
\end{proof}

\subsection{Proof of Proposition~\ref{prop:pers_equiv}}

\begin{proof}
Let us prove by induction that $Z_i^\mu = Y_i^\mu$ for all $ i \in \N$.
Base case: by definition, \(Z^\mu_0 = Y^\mu_0 = \emptyset \times U_d\).
Suppose \(Z^\mu_i = Y^\mu_i \) for some \(i \ge 1\). From Lemma~\ref{Lemma_persistence} we have $Z^\mu_{i+1} = (\Pre_{\Delta_{d}}(Z^\mu_{i+1}) \cap W) \cup \Pre_{\Delta_{d}}(Z^\mu_i)$
and  $Y^\mu_{i+1} = (\Pre_{ \Delta_{d, \varepsilon}}(Y^\mu_{i+1}) \cap W) \cup \Pre_{ \Delta_{d, \varepsilon}}(Z^\mu_i)$. Since $Z^\mu_j \in L^\mu$ for all $j\in \N$ and $\sup_{z\in Q(x)}\varepsilon(z,v) < \gamma(Q(x), v, L^\mu)$ for all $(x,v)$, Theorem~\ref{Thrm:Pre_mu} implies:
\( \Pre_{\Delta_d}(Z^\mu_i) = \Pre_{ \Delta_{d, \varepsilon}}(Z^\mu_i)\) and $\Pre_{\Delta_d}(Z^\mu_{i+1}) = \Pre_{ \Delta_{d, \varepsilon}}(Z^\mu_{i+1})$. Thus
\begin{equation}
\label{Z_mu}
Z^\mu_{i+1} = \big( \Pre_{ \Delta_{d, \varepsilon}}(Z^\mu_{i+1}) \cap W \big) \cup \Pre_{ \Delta_{d, \varepsilon}}(Z^\mu_i)
\end{equation}
Both $Z^\mu_{i+1}$ and $Y^\mu_{i+1}$ satisfy the fixed-point
equation $S = (\mathrm{Pre}_{ \Delta_{d, \varepsilon}}(S) \cap W)
\cup A_i$, where $A_i := \mathrm{Pre}_{ \Delta_{d, \varepsilon}}(Z^\mu_i)$.
Define
\begin{align*}
  \Phi_\varepsilon(S) &:= \bigl(\mathrm{Pre}_{ \Delta_{d, \varepsilon}}(S)
    \cap W\bigr) \cup A_i, \\
  \Phi(S) &:= \bigl(\mathrm{Pre}_{\Delta_d}(S) \cap W\bigr) \cup A_i.
\end{align*}
Since $\mathrm{Pre}_{ \Delta_{d, \varepsilon}}(S) \subseteq
\mathrm{Pre}_{\Delta_d}(S)$ for all $S \subseteq X_d$, it follows
that $\Phi_\varepsilon(S) \subseteq \Phi(S)$ for all
$S \subseteq X_d$, so Lemma~\ref{lemma_inclusion_oprators} gives
$\nu S.\,\Phi_\varepsilon(S) \subseteq \nu S.\,\Phi(S)$,
i.e.\ $Y^\mu_{i+1} \subseteq Z^\mu_{i+1}$.
For reverse inclusion, note that $Z^\mu_{i+1} =
\nu S.\,\Phi(S)$ is a fixed-point of $\Phi_\varepsilon$
(by~\eqref{Z_mu}), so the maximality of
$Y^\mu_{i+1} = \nu S.\,\Phi_\varepsilon(S)$ yields
$Z^\mu_{i+1} \subseteq Y^\mu_{i+1}$. Therefore $Z^\mu_{i+1} = Y^\mu_{i+1}$. By induction, $Z^\mu_i = Y^\mu_i$ for all $i \geq 0$, and thus $n = m$. Finally, since $Z^\mu_i = Y^\mu_i$ for all $i \in \{0,1,\ldots,n\}$, we conclude that $L^\mu = L^\mu_\varepsilon$. Consequently, for any $q \in X_d$, the index function $j_{L^\mu}(q)$ satisfies $j_{L^\mu}(q) = j_{L^\mu_{\varepsilon}}(q)$. Therefore, for any $q \in X_d$ with $j_{L^\mu}(q) < \infty$, we have:
\(
\mathcal{C}_d(q) = \mathcal{C}_{d,\varepsilon}(q) = \{ v \in U_d \mid (q,v) \in Z^\mu_{j_{L^\mu}}(q) = Y^\mu_{j_{L^\mu_{\varepsilon}}}(q) \}.
\) This completes the proof.
\end{proof}

\subsection{Proof of Theorem~\ref{thm:pers_refine}}

\begin{proof}
Consider a map \(\varepsilon: X \times U_d \rightarrow \mathbb{R}_{\geq 0}\) satisfying $\sup_{z\in Q(x)}\varepsilon(z,v) < \gamma(Q(x), v, L^\mu)$ for all \((x, u) \in X \times U_d\) and the symbolic model of the perturbed system \(S_{d}(\Sigma_{\varepsilon})=(X_d, X_d^0, U_d,  \Delta_{d, \varepsilon})\). Let $\mathcal{C}_{d, \varepsilon}$ the persistence controller for \(S_{d}(\Sigma_{\varepsilon})\) for the target set $X_T$. From Proposition~\ref{prop:pers_equiv} we have $\mathcal{C}_d = \mathcal{C}_{d, \varepsilon}$. Hence, it follows from Proposition \ref{Prop:refinement} that the controller \(\mathcal{C}\) is a persistence controller for the perturbed system \(S(\Sigma_{\varepsilon})\) and the persistence specification $Q^{-1}(\mathcal{H}_{Pers}^{X_T})$.
\end{proof}

\subsection{Proof of Proposition~\ref{Prop:pers_caracterization}}

\begin{proof}
From the definition of the non-shrinking iterations of the persistence specification, we have $Z^\mu_0 \subseteq Z^\mu_{1} \dots \subseteq Z^\mu_n$. Moreover, since the $\Pre_{\Delta_d}$ operator is monotone, we have $  \Pre_{\Delta_d}(Z^\mu_0) \subseteq \Pre_{\Delta_d}(Z^\mu_{1}) \dots \subseteq \Pre_{\Delta_d}(Z^\mu_n).$
Let  $(x, v) \in Q^{-1}(Z^\mu_{j_{L^\mu}(q)})$, $q = Q(x)$ and $j = j_{L^\mu}(q)$. We consider three mutually exclusive cases.
\noindent Case 1: $j = 1$. 
Then $(q,v) \in Z^\mu_1 = \Pre_{\Delta_d}(Z^\mu_1) \cap W$. 
By definition, $\Delta_d(q,v) \subseteq Z^\mu_1$ which is the smallest set in the list $L^\mu$, hence $I(q,v,L^\mu) = Z^\mu_1$.
\noindent Case 2: $j > 1$ and $(q,v) \notin \Pre_{\Delta_d}(Z^\mu_{j-1})$. 
Since $(q,v) \in Z^\mu_j = \big(\Pre_{\Delta_d}(Z^\mu_j) \cap W\big) \cup \Pre_{\Delta_d}(Z^\mu_{j-1})$, 
the assumption implies $(q,v) \in \Pre_{\Delta_d}(Z^\mu_j) \cap W $.
From $(q,v) \in \Pre_{\Delta_d}(Z^\mu_{j})$ and $(q,v) \notin \Pre_{\Delta_d}(Z^\mu_{j-1})$, 
we conclude $I(q,v,L^\mu) = Z^\mu_j$.
\noindent Case 3: $j > 1$ and $(q,v) \in \Pre_{\Delta_d}(Z^\mu_{j-1})$. 
By minimality of $j = j_{L^\mu}(q)$, we have $(q,v) \notin Z^\mu_{j-1}$.
Since $Z^\mu_{j-1} = \big(\Pre_{\Delta_d}(Z^\mu_{j-1}) \cap W\big) \cup \Pre_{\Delta_d}(Z^\mu_{j-2})$, it must hold that $(q,v) \notin \Pre_{\Delta_d}(Z^\mu_{j-2})$. From $(q,v) \in \Pre_{\Delta_d}(Z^\mu_{j-1})$ and $(q,v) \notin \Pre_{\Delta_d}(Z^\mu_{j-2})$, 
we conclude $I(q,v,L^\mu) = Z^\mu_{j-1}$.
All possible cases are covered, which completes the proof.
\end{proof}

\subsection{Proof of Lemma~\ref{lemma_inclusion_oprators}}

\begin{proof}
We proceed by induction on $i$. For $i = 0$, by definition, $Z_0^\varepsilon = X_d \times U_d = Z_0$, so $Z_0^\varepsilon \subseteq Z_0$ holds trivially. Assume $\nu^i Z.\, G_\varepsilon(Z) \subseteq \nu^i Z.\, G(Z)$ for some $i \geq 0$, we have $\nu^{i+1} Z.\, G_\varepsilon(Z) = G_\varepsilon\!\left(\nu^i Z.\, G_\varepsilon(Z)\right)
    \subseteq G_\varepsilon\!\left(\nu^i Z.\, G(Z)\right)  
    \subseteq G\!\left(\nu^i Z.\, G(Z)\right)  
    = \nu^{i+1} Z.\, G(Z), $
where the first inclusion comes from the induction hypothesis and the second inclusion from hypothesis $G_\varepsilon \subseteq G$ point-wise. This completes the induction. We can deduce that $\nu Z . G_\varepsilon(Z) \subseteq \nu Z . G(Z)$ by taking a large number of iterations in which both fixed-points are reached.
\end{proof}

\subsection{Proof of Lemma~\ref{Lemma_persistence}}

\begin{proof}
At the $i$-th iteration of the outer loop in Algorithm~\ref{alg:persistence}, we have from line~9 that $(Z^\mu_{i+1} = Z^\nu_j$), where $j$ denotes the final iteration index of the inner fixed-point computation. Combining this with line~8 yields \(Z^\mu_{i+1} = Z^\nu_j = Z^\nu_{j+1}.\) Hence, replacing both $Z^\nu_j$ and $Z^\nu_{j+1}$ by $Z^\mu_{i+1}$ in line $7$ provides the required results.
\end{proof}

\subsection{Proof of Proposition~\ref{prop:recu_equiv}}

\begin{proof}
The proof proceeds by induction on the iteration steps of Algorithm~\ref{alg:recurrence}. We have $Z_0^\nu = X \times U$ for both the nominal and perturbed systems, so $Z_0^\nu = Y_0^\nu$. Assume that for some $i \geq 0$, we have $Z_i^\nu = Y_i^\nu$. Consider the inner loop for outer index $i$ in Algorithm~\ref{alg:recurrence}. We prove by induction on $j$ that for all $j$, $Z_j^{\mu,i} = Y_j^{\mu,i}$. \textbf{Base case ($j=0$):} $Z_0^{\mu,i} = \emptyset \times U = Y_0^{\mu,i}$. \textbf{Induction step:} Assume $Z_j^{\mu,i} = Y_j^{\mu,i}$.
From Algorithm~\ref{alg:recurrence}, the update is:\(
Z_{j+1}^{\mu,i} = \big( \Pre_{\Delta_d}(Z_i^\nu) \cap W \big) \cup \Pre_{\Delta_d}(Z_j^{\mu,i}).
\) Since $Z_i^\nu \in L^\mu$ and $Z_j^{\mu,i} \in L^\mu$, and $\sup_{z\in Q(x)}\varepsilon(z,v) < \gamma(Q(x), v, L^\mu)$ for all $(x,v)$, Theorem~\ref{Thrm:Pre_mu} implies:
\(
\Pre_{\Delta_d}(Z_i^\nu) = \Pre_{ \Delta_{d, \varepsilon}}(Y_i^\nu) \quad \text{and} \quad \Pre_{\Delta_d}(Z_j^{\mu,i}) = \Pre_{ \Delta_{d, \varepsilon}}(Y_j^{\mu,i}).
\)
Therefore, \(
    Z_{j+1}^{\mu,i} = \big( \Pre_{ \Delta_{d, \varepsilon}}(Y_i^\nu) \cap W \big) \cup \Pre_{ \Delta_{d, \varepsilon}}(Y_j^{\mu,i}) = Y_{j+1}^{\mu,i}.
    \)
Thus, by induction on $j$, we have $Z_j^{\mu,i} = Y_j^{\mu,i}$ for all $j$, including the fixed-point $j = m(i)$. The outer update sets $Z_{i+1}^\nu = Z_{m(i)}^{\mu,i}$ and $Y_{i+1}^\nu = Y_{m(i)}^{\mu,i}$. Since $Z_{m(i)}^{\mu,i} = Y_{m(i)}^{\mu,i}$, we conclude $Z_{i+1}^\nu = Y_{i+1}^\nu$. By induction on $i$, we have $Z_i^\nu = Y_i^\nu$ for all $i$, including the final fixed-point $i = n$. Therefore, $Z_n^\nu = Y_n^\nu$ and $n_1 = n$, which implies: $\mathrm{dom}(\mathcal{C}_d) = \pi_X(Z_n^\nu) = \pi_X(Y_n^\nu) = \mathrm{dom}(\mathcal{C}_{d,\varepsilon}).$ Now consider \( q \in X_d \), since $L^\mu = L^\mu_\varepsilon$ , it follows that \( L = L_{\varepsilon} \), which implies that \( j_{L}(q) = j_{L_{\varepsilon}}(q) \). Therefore, for any \( q \in X_d \) with \( j_{L}(q) < \infty \), we have
\( \mathcal{C}_d(q) = \mathcal{C}_{d, \varepsilon}(q) = \{ v \in U \mid (q,v) \in  Z_{j_{L}}(q) =  Y_{j_{L_{\varepsilon}}}(q)\} .\)
\end{proof}

\subsection{Proof of Theorem~\ref{thm:recu_refine}}

\begin{proof}
Consider a map \(\varepsilon: X \times U_d \rightarrow \mathbb{R}_{\geq 0}\) satisfying {$\sup_{z\in Q(x)}\varepsilon(z,v) < \gamma(Q(x), v, L^\mu)$} and the symbolic model of the perturbed system \(S_{d}(\Sigma_{\varepsilon})=(X_d, X_d^0, U_d,  \Delta_{d, \varepsilon})\). From Proposition~\ref{prop:recu_equiv} where $\mathcal{C}_{d, \varepsilon}$ is a recurrence controller for \(S_{d}(\Sigma_{\varepsilon})\) and the same target set $X_T$, we have $\mathcal{C}_d = \mathcal{C}_{d, \varepsilon}$. Hence, it follows from Proposition \ref{Prop:refinement} that the controller \(\mathcal{C}\) is a recurrence controller for the perturbed system \(S(\Sigma_{\varepsilon})\) and the recurrence specification $Q^{-1}(\mathcal{H}_{Rec}^{X_T})$.
\end{proof}

\subsection{Proof of Proposition~\ref{prop:unified_comparison}}

\begin{proof}
{Let $(x,v) \in Q^{-1}(Z_\infty)$ and set $q = Q(x)$. By the definition of
$\gamma$ in~\eqref{Def:gamma_l}, we have
$\gamma(q,v,L_\Phi) = \eta(q,v,I(q,v,L_\Phi))$, so by the monotonicity of
$\eta$ (Lemma~\ref{lem:2}) it suffices to establish
$\Delta_d(q,v) \subseteq I(q,v,L_\Phi)$. We first verify that the family defining $I(q,v,L_\Phi)$ is non-empty by
showing $(q,v) \in \mathrm{Pre}_{\Delta_d}(Z_\infty)$ equivalent to
$\Delta_d(q,v) \subseteq \pi_X(Z_\infty)$, for each
specifications. Indeed, for safety,
$Z_\infty = \Pre_{\Delta_d}(Z_\infty) \cap W \subseteq
\Pre_{\Delta_d}(Z_\infty)$ gives $\Delta_d(q,v) \subseteq \pi_X(Z_\infty)$;
for persistence, \(Z_\infty = \big( \Pre_{\Delta_d}(Z_{\infty}) \cap W \big) \cup \Pre_{\Delta_d}(Z_\infty)  =  \Pre_{\Delta_d}(Z_\infty)\) gives
$\Delta_d(q,v) \subseteq \pi_X(Z_\infty)$; for recurrence, $Z_\infty
= \big( \Pre_{\Delta_d}(Z_\infty) \cap W \big) \cup \Pre_{\Delta_d}(Z_\infty)
= \Pre_{\Delta_d}(Z_\infty)$  gives
$\Delta_d(q,v) \subseteq \pi_X(Z_\infty)$; and for reachability,
$(q,v) \in Z_\infty = \Pre_{\Delta_d}(Z_\infty) \cup W$ and $(q, v) \notin W$ gives
$\Delta_d(q,v) \subseteq \pi_X(Z_\infty)$, since $(q, v) \in W = X_T \times U_d$ means we already reach the target and no succession required. All the sets $\pi_X(Z_i)$ considered for the intersection $I(q,v,L_\Phi)$ contain $\Delta_d(q,v)$ which implies $\Delta_d(q,v) \subseteq I(q,v,L_\Phi)$. }
\end{proof}

\end{document}